\documentclass[11pt]{article}

\usepackage{authblk}
\usepackage{amssymb,latexsym,amsmath,amsthm,enumitem}
\usepackage{stmaryrd}
\usepackage{amsfonts}

\theoremstyle{definition}
\newtheorem{theorem}{Theorem}
\newtheorem{corollary}[theorem]{Corollary}
\newtheorem{proposition}[theorem]{Proposition}
\newtheorem{lemma}[theorem]{Lemma}
\newtheorem{definition}[theorem]{Definition}
\newtheorem{example}[theorem]{Example}
\newtheorem{notation}[theorem]{Notation}
\newtheorem{remark}[theorem]{Remark}

\usepackage{tikz}
\usepackage[all,cmtip]{xy}
\usepackage{hyperref}
\usepackage{accents}
\usepackage[margin=2.5cm]{geometry}

\usepackage{titling}
\hypersetup{
  colorlinks   = true, 
  urlcolor     = blue, 
  linkcolor    = blue, 
  citecolor   = red 
}

\newcommand{\numberset}{\mathbb}
\newcommand{\N}{\numberset{N}}
\newcommand{\Z}{\numberset{Z}}

\newcommand{\R}{\numberset{R}}
\newcommand{\C}{\numberset{C}}

\newcommand{\F}{\numberset{F}}

\newcommand{\mC}{\mathcal{C}}
\newcommand{\mP}{\mathcal{P}}

\newcommand{\mL}{\mathcal{L}}
\newcommand{\mA}{\mathcal{A}}

\newcommand{\mD}{\mathcal{D}}

\newcommand{\mQ}{\mathcal{Q}}

\newlength{\dhatheight}
\newcommand{\doublehat}[1]{%
\settoheight{\dhatheight}{\ensuremath{\hat{#1}}}%
\addtolength{\dhatheight}{-0.37ex}%
\hat{\vphantom{\rule{2pt}{\dhatheight}}%
\smash{\hat{#1}}}}

\title{\textbf{Duality of Codes Supported on    
Regular Lattices, with an
Application to Enumerative Combinatorics}}

\author{Alberto Ravagnani\thanks{E-mail:
\texttt{ravagnani@ece.utoronto.ca}. The author
was partially supported by the Swiss National Science
Foundation through grant no. 200021\_150207. }}
  
 \affil{Institut de Math\'{e}matiques \\ Universit\'{e} de
Neuch\^{a}tel \\ Emile-Argand 11, CH-2000 Neuch\^{a}tel, Switzerland}

\date{}

\begin{document}

\maketitle

\begin{abstract}
We introduce a general class of regular weight functions on finite abelian groups, 
and study the combinatorics, the duality theory, and the metric properties of codes 
endowed with such functions. The weights are obtained by 
composing a suitable support map with the rank function of a graded lattice satisfying certain 
regularity properties. 
A regular weight on a group canonically induces a regular weight on the character group, and  
invertible MacWilliams identities always hold for such a pair of weights. 
Moreover, the  Krawtchouk coefficients of the corresponding MacWilliams transformation have a  
precise combinatorial significance, and can be expressed in terms of the invariants of the underlying
 lattice. In particular, they are easy to compute in many
examples. Several weight functions traditionally studied in Coding Theory belong to the class 
of weights introduced in this paper. Our lattice-theory approach also offers a control on metric 
structures that a regular weight 
induces on the underlying group. In particular, it allows to show that every finite abelian group admits 
weight functions that, simultaneously, 
give rise to MacWilliams identities, and endow the underlying group with a metric space structure. 
We propose a general notion of extremality for (not necessarily additive) 
codes in groups endowed with  
semi-regular supports, and establish a  
 Singleton-type bound. We then investigate the combinatorics
and duality theory of extremal codes, extending
 classical results on the weight and distance distribution of 
error-correcting codes. Finally, we apply the theory of MacWilliams identities 
 to enumerative combinatorics problems, obtaining closed formul{\ae} for the number of rectangular matrices 
over a finite having prescribed rank and satisfying some linear conditions. 
\end{abstract}

\vspace{1.0em}

\section*{Introduction and motivations} \label{intr} 
 In Coding Theory, a MacWilliams identity expresses a linear
 transformation between the partition enumerator
 of a code and the  partition enumerator of its dual code. MacWilliams
 identities are named after
 F. J. MacWilliams, who first discovered in \cite{ooo} relations of this type for linear codes endowed 
 with the Hamming weight. Analogous identities were later established for other classes
 of codes, most notably for  codes in groups (see \cite{by, forney, heide1, MS, zino0, zino} 
 and the references therein). 

An additive code
$\mC\subseteq G$ 
is a subgroup of a finite abelian group $G$, and its dual code $\mC^* \subseteq \hat{G}$ is
 its character-theoretic annihilator. Codes and dual codes are
 subsets of different ambient spaces, and therefore their enumerators
 refer in general to different partitions, say $\mP$ and $\mQ$,
on $G$ and $\hat{G}$, respectively. When 
$\mP$ and $\mQ$ are ``mutually compatible'', the 
$\mP$-distribution of a code $\mC$ and the $\mQ$-distribution of the dual code $\mC^*$
determine each other. The linear relations between the partition distributions are 
expressed by certain complex numbers called Krawtchouk coefficients. Their existence is
guaranteed by the compatibility property of the partitions, but giving explicit formul{\ae}
for them is difficult in general.

Most general works on MacWilliams identities for codes in groups focus 
on group partitions and their duals, as these determine 
the existence of a MacWilliams transformation.  More precisely, it is known that 
a partition $\mP$ of a finite abelian group $G$
induces a dual partition $\hat{\mP}$ of the character group $\hat{G}$.
Under a certain ``compatibility'' assumption (called Fourier-reflexivity), the partition 
enumerators of a code and of its dual code associated to  
$\mP$ and $\hat{\mP}$, respectively, determine each other via a MacWilliams 
transformation (see \cite{heide1,zino0,zino} for details). Compatible pairs were  
proposed in~\cite{zino} with the goal of simplifying the construction 
of abelian association schemes on groups.

In Coding Theory however, partitions of the ambient space are generally induced by  
 weight functions having an information-theoretic significance. 
 The main property that usually allows error correction 
is the \textit{weight} function defined on the ambient group, rather 
than the \textit{partition} it induces.
A given group partition can be induced  by many different weights, which in general  
will not have good properties from a Coding Theory viewpoint. 
For example, they will not endow the underlying group with a metric structure. 
Weight functions and induced partitions are therefore not 
equivalent information in general.

In this paper we focus our attention on weight functions, and address
 the problem of constructing \textit{general} families of \textit{numerical} 
 weights on finite abelian groups (rather than partitions) that yield MacWilliams 
 identities. We achieve this combining lattice theory and discrete Fourier 
 analysis methods, introducing a general notion of support map. 
One of the advantages of our approach is that it offers a control (via lattice modularity) 
on metric space structures induced by the weight function on the ambient group.  
 In a second part of this work we investigate how the theory of weights on groups relate to codes'
extremality. We close the paper with a section devoted to 
enumerative combinatorics problems on matrices.
More in detail, this paper makes the following contributions.

In Section \ref{secm} we define a regular support $\sigma$ as a function 
on a finite abelian group 
$G$
that takes values in a graded lattice $\mL$ with certain regularity properties. A regular support 
naturally induces a weight function on 
$G$ via the rank function of $\mL$. 
We show that a regular support $\sigma$ on $G$ \textit{canonically} induces a dual regular 
support $\sigma^*$ on the character group $\hat{G}$, with values in the dual lattice $\mL^*$. 
This defines in particular a {canonical} and numerical weight on $\hat{G}$ via the rank function 
of $\mL^*$. In contrast to previous works on codes in groups, our approach concentrates on 
numerical weights on the groups $G$ and $\hat{G}$, rather than on partitions of them. 

In Section \ref{secmw} we show that the weight functions induced by a regular support and its 
dual support, respectively, always obey a MacWilliams identity. The result relies on the 
specific notion of dual support, as the definition of regular support on a group does not depend 
on the algebraic structure of the corresponding character group. 
The Krawtchouk coefficients of the MacWilliams 
transformation are integers
with a precise combinatorial significance. More precisely, they can be expressed in terms of the 
combinatorial invariants of the  support lattice $\mL$.

 As a secondary result, in Section \ref{recover} we revisit the theory of MacWilliams identities 
associated with many weights traditionally studied in Coding Theory, showing that such weight 
functions factor through a suitable support map having  remarkable regularity
 properties.

In Section \ref{secmetric} we show that when the lattice associated with a regular support is modular, then the underlying group can be naturally endowed with a metric space structure.  We then prove that one can systematically construct (over any finite abelian group) numerical weight functions
that, \textit{simultaneously}, yield MacWilliams identities, and endow the ambient 
group with a metric structure. This is particularly interesting from a Coding Theory perspective, as the triangle inequality is usually a key property enabling error correction.

In Section \ref{secop} we propose a general notion of extremality for codes in groups endowed with semi-regular support functions, establishing a generalized Singleton bound in this context. Our definition of extremality only relies on codes' cardinality, which is a fundamental code parameter from an information-theoretic viewpoint. This differs from previous general approaches, where extremal codes are defined in terms of their interaction with the dual code via algebraic properties of their inner distributions.  We also study the combinatorics and duality theory of the codes attaining the Singleton-type bound, extending classical results to the general framework of codes in groups. 

In Section \ref{enmat} the theory of MacWilliams identities is applied to enumerative combinatorics problems, deriving explicit formul{\ae} for the number of $k \times m$ matrices over $\F_q$ having prescribed rank and satisfying certain linear conditions (e.g., having zeroes in a given set of diagonal entries). In particular, we answer a generalized question of R. Stanley.


\section{Groups, Codes, and Weight Functions} \label{sec3}

Let $(G,+)$ be an abelian group. The \textbf{character group} of $G$,
denoted by $(\hat{G},\cdot)$, is the set of group homomorphisms $\chi:G
\to \C^*=\C \setminus \{0\}$ endowed with pointwise multiplication, i.e., for
$\chi_1, \chi_2 \in \hat{G}$,
$$(\chi_1 \cdot \chi_2)(g):=\chi_1(g) \chi_2(g), \ \ \ \mbox{ for all } g \in
G.$$ The neutral element of $(\hat{G},\cdot)$ is the  \textbf{trivial character}
$\varepsilon \equiv 1$ of $G$.
The
groups 
$G$ and $\doublehat{G}$ are canonically isomorphic via the map
$\psi: G \to \doublehat{G}$ defined, for $g \in G$, by $\psi(g)(\chi):=\chi(g)$
for all
$\chi \in \hat{G}$. It is well-known that when $(G,+)$ is finite and abelian
the groups
$(G,+)$ and
$(\hat{G},\cdot)$ are isomorphic, not canonically in general. In particular,
$|G|=|\hat{G}|$. Notice  that for all $n \ge 1$ we have
$\widehat{G^n}=\hat{G}^n$, where $(\chi_1,...,\chi_n) \in \widehat{G^n}$ is 
defined,
for all $(g_1,...,g_n) \in G^n$,
by
 $$(\chi_1,...,\chi_n)(g_1,...,g_n):= \prod_{i=1}^n \chi_i(g_i).$$

\begin{definition}
Let $G$ be a finite abelian group. A
\textbf{code} in $G$ is a subgroup 
$\mC \subseteq G$.  The \textbf{dual} of 
$\mC$
 is the code  $\mC^*:=\{ \chi \in \hat{G} :
\chi(g)=1 
\mbox{ for all } g \in \mC\} \subseteq \hat{G}$.
We say that $\mC$ is \textbf{trivial} if $\mC=\{0\}$
or $\mC=G$. The code
\textbf{generated} by codes $\mC,\mD \subseteq G$ is the
 code  $\mC+\mD:= \{ c+d : c \in \mC, \ d \in \mD\}
\subseteq G$.
\end{definition}

The following remark summarizes some properties of duality.
The proof is left to the reader.

\begin{remark} \label{cardual}
 Let
$\mC \subseteq G$ be a code. Then  
$|\mC| \cdot |\mC^*|=|G|=|\hat{G}|$. Moreover, identifying $G$ and
$\doublehat{G}$  we have $\mC^{* *}=\mC$. Finally, duality and sum of
codes relate as follows.
\begin{enumerate}
  \item Let $\mC,\mD \subseteq G$ be   codes.
Then $|\mC + \mD| \times |\mC \cap \mD| =|\mC| \cdot |\mD|$. \label{dimsuma}

\item Let $\mC_1,...,\mC_t \subseteq G$ be   codes, $t \ge 2$.
We have $\bigcap_{i=1}^t \mC_i^* = \left( \sum_{i=1}^t \mC_i
\right)^*$. \label{dimsumb}
 \end{enumerate}
\end{remark}

\begin{definition}
 Let $G$ be a finite abelian group. A \textbf{weight} on $G$ is a
function
$\omega :G \to X$, where $X$ is a finite non-empty set. 
The \textbf{$\omega$-distribution} of a code
$\mC \subseteq G$ is the integer vector
$(W_a(\mC,\omega) : a \in X)$, where $W_a(\mC,\omega):=|\{ g \in \mC : \omega(g)=a\}|$
for all $a \in X$.

A weight function on a group naturally induces a partition of it as follows.
\begin{definition} \label{equiva}
 Let $\omega:G \to X$ be a weight.
For all $a\in \omega(G)$ define $P_a(\omega):=\{ g \in G : \omega(g)=a\}$.
Then
$$\mP(\omega):= \bigsqcup_{a \in \omega(G)} P_a(\omega)$$ is 
the \textbf{partition} of
$G$ 
induced
by $\omega$. We say that weight functions $\omega:G \to X$ and $\omega':G \to
X'$ are
\textbf{equivalent} if $\mP(\omega)=\mP(\omega')$, and in this case we write
$\omega \sim \omega'$.
\end{definition}

Let 
$\omega:G \to X$ and $\tau:\hat{G} \to Y$ be weights. We say that
$(\omega,\tau)$ is a \textbf{compatible pair} if for all $b \in \tau(\hat{G})$ and 
for all $g,g' \in G$ with $\omega(g)=\omega(g')$
one has
$$\sum_{\chi \in P_b(\tau)} \chi(g) \ \ = \sum_{\chi \in P_b(\tau)} \chi(g').$$
If this is the case, then
the \textbf{Krawtchouk
coefficients} associated to $(\omega,\tau)$ are defined, for every
 $a \in \omega(G)$ and $b \in \tau(\hat{G})$,  by 
$$K(\omega,\tau)(a,b):=\sum_{\chi \in P_b(\tau)}
\chi(g),$$
where $g \in G$ is any element with $\omega(g)=a$.
When $a \notin \omega(G)$ or $b \notin \tau(\hat{G})$ we put
$K(\omega,\tau)(a,b):=0$.
\end{definition}

\begin{remark} \label{dualkk}
 Let $\omega:G \to X$, $\tau:\hat{G} \to Y$ be weights.
Identifying $G$ and $\doublehat{G}$ one has $g(\chi)=\chi(g)$ for all
$g \in G$ and $\chi \in \hat{G}$. Thus when $(\tau,\omega)$ is a compatible
pair the Krawtchouk
coefficients associated to $(\tau,\omega)$ are defined, for every 
$a \in \tau(\hat{G})$ and $b \in \omega(G)$, by
$$K(\tau,\omega)(a,b)=\sum_{g \in P_b(\omega)} \chi(g),$$
where $\chi \in \hat{G}$ is any character with $\tau(\chi)=a$.
Again, if  $a \notin \tau(\hat{G})$ or $b \notin \omega(G)$ then 
$K(\tau,\omega)(a,b)=0$.
\end{remark}

\begin{remark} \label{distf}
 Let $\omega:G \to X$, $\omega':G \to X'$, $\tau:\hat{G} \to Y$ and 
 $\tau':\hat{G} \to Y'$ be
weights with
$\omega \sim \omega'$ and $\tau \sim \tau'$. There exist bijections
$\pi:\omega'(G) \to \omega(G)$ and $\eta:\tau'(\hat{G}) \to \tau(\hat{G})$ such
that $\omega=\pi \circ \omega'$ and $\tau=\eta \circ \tau'$.
Moreover, it is easy to see that if $(\omega,\tau)$
is a compatible pair, 
then $(\omega',\tau')$ is also a compatible pair, and
for all $a \in \omega'(G)$ and 
$b \in \tau'(\hat{G})$ one has
$$K(\omega',\tau')(a,b)=K(\omega,\tau)(\pi(a),\eta(b)).$$
Therefore the Krawtchouk
coefficients
associated to $(\omega',\tau')$ are essentially the same as the Krawtchouk
coefficients
associated to $(\omega,\tau)$, up to a suitable bijection. For this reason most
authors concentrate on group partitions 
when studying Krawtchouk coefficients in the context of additive codes.

In Coding Theory however, given a ``numerical'' weight function $\omega:G \to
X \subseteq \N$, one
naturally attempts to define a distance $d_\omega$ on $G$
by setting
$d_\omega(g,g'):=\omega(g-g')$ for all $g,g' \in 
G$. This is usually crucial for error correction.
It is easy to
construct groups $G$ and weights
$\omega,\omega':G \to X \subseteq \N$ such that
$\omega \sim \omega'$, $d_\omega$ is a distance function, but
$d_{\omega'}$ is not. This is the reason why in this work we  concentrate 
on numerical weights, rather than on group partitions. 
\end{remark}

It is well known~\cite[Theorem 1]{zino0} that compatible pairs of 
weights yield MacWilliams-type identities as follows.

\begin{theorem}[MacWilliams Identities]\label{mwwide}
 Let $G$ be a finite abelian group, and let
$\omega:G \to X$ and $\tau:\hat{G} \to Y$ be weights.
Assume that $(\omega,\tau)$ is compatible. Then
for all codes $\mC \subseteq G$  we have
$$W_b(\mC^*,\tau)=\frac{1}{|\mC|}\sum_{a \in X} K(\omega,\tau)(a,b)
 \ W_a(\mC,\omega).$$
for all $b \in Y$. In particular, the $\omega$-distribution of $\mC$
determines the
$\tau$-distribution of $\mC^*$.
\end{theorem}

\begin{proof}
 By the definition of compatible pair, the partition $\mP(\omega)$ 
 is finer than
 $\widehat{\mP(\tau)}$, the dual of the partition $\mP(\tau)$
 (see \cite[Definition 2.1]{heide1}). 
 The result now follows from \cite[Theorem 2.7]{heide1}.
\end{proof}

Notice that the identities of Theorem \ref{mwwide} express a linear 
transformation between the $\omega$-distribution of the code $\mC \subseteq G$
and the $\tau$-distribution of its dual code $\mC^* \subseteq \hat{G}$. 
The matrix of the linear transformation, $K(\omega,\tau)$, is called the 
\textbf{Krawtchouk matrix}. Its rows are indexed by 
$b \in Y$, and its columns are indexed by $a \in X$.

\begin{remark} \label{probs}
The fact that a pair  $(\omega,\tau)$ is compatible does not
imply  
that $(\tau,\omega)$ is compatible. 
This corresponds to the fact that the MacWilliams transformation is not
invertible.
The most interesting scenario is when both $(\omega,\tau)$ and
$(\tau,\omega)$ are compatible,  i.e., when $\omega$ and $\tau$ are
\textbf{mutually}
compatible. 
\end{remark}

We conclude this section by mentioning the product weight and the
symmetrized weight induced by a weight function. See \cite{heide1,zino0} 
for a more complete analysis.

\begin{definition}
 Let $\omega:G \to X$ be a weight, and let $n \ge 1$ be an integer. 
\begin{enumerate}
 \item The \textbf{product weight} on $G^n$ associated to $\omega$ is
the function
$\omega^n:G^n \to X^n$ defined, for all $(g_1,...,g_n)$, by 
$\omega^n(g_1,...,g_n):=(\omega(g_1),...,\omega(g_n))$.
\item Assume that $X=\{0,...,r\}$
and for all $(c_1,...,c_n) \in X^n$ let 
$\mbox{cmp}(c):=(e_0,...,e_r)$, where $e_i:=|\{ 1 \le j \le n : c_j=x_i\}|$
for all $0 \le i \le r$.
The \textbf{symmetrized weight} on $G^n$ associated to $\omega$ is the function
$\omega^n_{\textnormal{sym}}:G^n \to \{ 0,...,n\}^{r+1}$ defined, for all
$(g_1,...,g_n) \in G^n$, by 
$\omega^n_{\textnormal{sym}}(g_1,...,g_n):=\mbox{cmp}(\omega^n(g_1,...,g_n))$. 
\end{enumerate}
\end{definition}

Compatibility of pairs  is preserved by products and symmetrization. 

\begin{proposition} \label{ext}
Let $\omega:G \to X$ and $\tau:\hat{G} \to Y$ be weights. Let 
$n \ge 1$, $r:=|X|$ and $s:=|Y|$. Assume that
 $(\omega,\tau)$ is compatible. Write 
 $K=K(\omega,\tau)$ for ease of notation. The following hold.
 
 \begin{enumerate}
 \item The pair $(\omega^n,\tau^n)$ is compatible. Moreover,
for $a=(a_1,...,a_n) \in X^n$ and $b=(b_1,...,b_n)\in Y^n$
we have 
$$\displaystyle K({\omega^n,\tau^n})(a,b) \ = \ \prod_{j=1}^n
 K(a_j,b_j).$$
 
 \item Assume $X=\{0,...,r\}$ and $Y=\{0,...,s\}$.
 The pair $(\omega^n_{\textnormal{sym}},
\tau^n_{\textnormal{sym}})$,
is compatible. Moreover,
for $d=(d_0,...,d_r)\in \{1,...,n\}^{r+1}$ and 
$e \in \{1,...,n\}^{s+1}$ we
have
$$\displaystyle K({\omega^n_{\textnormal{sym}},\tau^n_{\textnormal{sym}}})(d,e)
 \ = \sum_{\substack{b \in 
Y^n \\ \textnormal{cmp}(b)=e}}
\prod_{j=1}^{d_0} K(0,b_j) \prod_{j=d_0+1}^{d_0+d_1}
K(1,b_j) \ \cdots \prod_{j=d_0 + \cdots + d_{r-1}+1}^{d_0 + \cdots + d_r} K(r,b_j).$$
 \end{enumerate}
 \end{proposition}

\begin{proof}
 Let $(a_1,...,a_n) \in \omega^n(G^n)$ and $(b_1,...,b_n) \in
\tau^n(\hat{G}^n)$. For any element $(g_1,...,g_n) \in G^n$
with
$\omega^n(g_1,...,g_n)=(a_1,...,a_n)$ one has
\begin{equation} \label{eqqq}
 \sum_{\substack{(\chi_1,...,\chi_n) \in \hat{G}^n \\
\tau^n(\chi_1,...,\chi_n)=(b_1,...,b_n)}}
(\chi_1,...,\chi_n)(g_1,...,g_n)=\prod_{j=1}^n K(a_j,b_j).
\end{equation}
This shows that $(\omega^n,\tau^n)$ is a compatible pair, and proves the
first formula in the statement. 
Now we study the symmetrized weight. Let $(d_0,...,d_r) \in
\omega^n_\textnormal{sym}(G^n)$ and $(e_0,...,e_s) \in
\tau^n_\textnormal{sym}(\hat{G}^n)$, and let $(g_1,...,g_n) \in G^n$
with $\omega^n_\textnormal{sym}(g_1,...,g_n)=(d_0,...,d_r)$. Using
(\ref{eqqq})
we compute
\begin{equation} \label{eqqqq}
 \sum_{\substack{(\chi_1,...,\chi_n) \in \hat{G}^n \\
\tau^n_{\textnormal{sym}}(\chi_1,...,\chi_n)=(e_0,...,e_s)}}
(\chi_1,...,\chi_n)(g_1,...,g_n)\ \ =\sum_{\substack{(b_1,...,b_n) \in Y^n \\
\textnormal{cmp}(b_1,...,b_n)=(e_0,...,e_s)}}\prod_{j=1}^n
K(a_j,b_j),
\end{equation}
where $(a_1,...,a_n):=\omega^n(g_1,...,g_n)$. Up to a permutation of the entries
of $(a_1,...,a_n)$, without loss of generality we may assume $a_i \le a_{i+1}$
for all $1 \le i \le n-1$.
Therefore (\ref{eqqqq}) becomes
$$\displaystyle \sum_{\substack{b \in 
Y^n \\ \textnormal{cmp}(b)=e}}
\prod_{j=1}^{d_0} K(0,b_j) \prod_{j=d_0+1}^{d_0+d_1}
K(1,b_j) \ \cdots \prod_{j=d_0 + \cdots + d_{r-1}+1}^{d_0 + \cdots + d_r} K(r,b_j).$$
The expression above only depends on $(d_0,...,d_r)$ and $(e_0,...,e_s)$. This 
shows that $(\omega^n_{\textnormal{sym}},
\tau^n_{\textnormal{sym}})$ is compatible, and proves the second
formula in the statement.
\end{proof}

 Proposition \ref{ext} shows  that 
the computation of the Krawtchouk
coefficients of the pairs  $(\omega^n,\tau^n)$ and
$(\omega^n_{\textnormal{sym}}, \tau^n_{\textnormal{sym}})$
reduces to the computation of the Krawtchouk coefficients of $(\omega,\tau)$. 

In the reminder of the paper we concentrate on numerical weight functions on 
finite abelian groups arising from lattices.


\section{Regular Lattices} \label{seclattic}

In this section we briefly recall some basic notions on posets and lattices,
and propose a definition of regular lattice.
See  \cite[Chapter 3]{ec} for a
general introduction to posets. Throughout this paper we only 
treat finite lattices.

Given a poset $(L, \le)$ and $S,T \in L$, we write $S < T$ for
$S \le T$ and $S \neq T$. We write $S \lessdot T$ if $S<T$ and
there is no $U \in L$ with $S<U<T$. In this case we say that $T$ \textbf{covers}
$S$. Recall moreover that a \textbf{meet} of $S,T \in L$ is a 
maximal lower bound for both $S$ and $T$. Similarly, a \textbf{join} of 
$S,T \in L$ is a 
minimal upper bound for both $S$ and $T$.

\begin{definition}
 A \textbf{lattice} is a poset $(L,\le)$ where every $S,T \in L$
have a unique meet
and a unique join, denoted by
$S \wedge T$ and $S\vee T$, respectively.
\end{definition}

Meet and join of a lattice $\mL=(L,\le,\wedge,\vee)$
 define two binary,
commutative and associative  operations
$\wedge, \vee:L \times L \to L$.
In particular, for any non-empty finite subset $M \subseteq L$, the lattice 
elements
$\bigwedge \{ S : S \in M\}$ and $\bigvee \{ S : S \in M\}$ are well defined.
When $\mL$ is \textbf{finite} (i.e., $L$ is finite), we set 
$0_\mL:= \bigwedge \{ S : S \in L\}$ and $1_\mL:= \bigvee \{ S : S \in L\}$.

A  finite lattice $\mL$ is \textbf{graded} of \textbf{rank} $r$ 
if all maximal chains (with respect to $\le$)
in $\mL$ have the same length $r$. We denote the rank of a graded lattice $\mL$
by
$\mbox{rk}(\mL)$.

\begin{remark} \label{rk_def}
 Let $\mL=(L,\le,\wedge,\vee)$ be a finite  graded lattice  of rank $r$.
 There exists a unique
function
$\rho_\mL:L \to \{0,...,r\}$, called the \textbf{rank
function}
of  $\mL$, with 
$\rho(0_\mL)=0$ and $\rho_\mL(T)=\rho_\mL(S)+1$ whenever $S \lessdot T$
(see \cite[page 281]{ec}). 
The function $\rho_\mL$ is monotonic, i.e., $\rho_\mL(S) \le
\rho_\mL(T)$
whenever $S \le T$. Moreover, $\rho_\mL(L)=\{0,...,r\}$, and $0_\mL$ 
and $1_\mL$ are the only
elements of rank $0$ and $r$, respectively.
\end{remark}

The \textbf{dual} of a lattice
$\mL=(L,
\le, \wedge, \vee)$ is the lattice
$\mL^*=(L, \preceq, \curlywedge, \curlyvee)$, where
$S \preceq T$ if and only if $T \le S$, $\curlywedge:= \vee$ and
$\curlyvee:= \wedge$.
 If $\mL$ is finite (and so $\mL^*$ is finite) then $0_{\mL^*}=1_\mL$
and
$1_{\mL^*}=0_\mL$. Clearly, $\mL^{**}=\mL$. Notice moreover that
 $\mL$ is graded if and only if $\mL$ is graded. If this is the case, then
$\mbox{rk}(\mL)=\mbox{rk}(\mL^*)$ and
$\rho_{\mL^*}(S)=\mbox{rk}(\mL)-\rho_\mL(S)$ for all $S \in L$.

\begin{definition}
 Let $\mL=(L, \le)$ be a finite poset. Then the \textbf{M\"{o}bius function} of 
$\mL$ is the map $\mu_\mL:\{(S,T) \in L \times L : S \le T \} \to
\Z$ inductively defined 
 by
$\mu_\mL(S,S)=1$ for all $S \in L$, and 
$$\mu_\mL(S,T)= -\sum_{S \le U<T} \mu_\mL(S,U) \ \ \mbox{ for all $S,T \in L$
with } S <T.$$
\end{definition}

Using  the fact that a lattice $\mL$ and its dual lattice $\mL^*$ 
    are anti-isomorphic, 
 one can show that
$\mu_{\mL^*}(S,T)=\mu_{\mL}(T,S)$ for all
$S,T \in \mL$ (see \cite[Proposition 2.1.10]{incal}).

Now we introduce regular lattices, which are central in our approach.

\begin{definition} \label{defmoreg}
 A finite graded lattice $\mL=(L,\le, \wedge, \vee)$
of rank $r$ is \textbf{regular} if the following hold.
\begin{enumerate}[label=(\alph*)] \setlength\itemsep{0.2em}
 \item For all \label{AAA}
$T \in L$ and for all integers $0 \le s \le r$,
\begin{itemize}  \setlength\itemsep{0.05em}
\item  the number of
$S \in L$ with $\rho_\mL(S)=s$ and $S \le T$ only depends on $s$ and
$\rho_\mL(T)$,
 \item  the number of
$S \in L$ with $\rho_\mL(S)=s$ and $T \le S$ only depends on $s$ and
$\rho_\mL(T)$.
\end{itemize}
\item For all $S,T \in L$ with $S \le T$, the M\"{o}bius 
function $\mu_\mL(S,T)$
only depends \label{BBB}
on $\rho_\mL(S)$
and $\rho_\mL(T)$.
\end{enumerate}
\end{definition}

A regular lattice is shown in Figure \ref{fig:1} via its Hasse diagram. More examples 
will be given 
in Section \ref{recover}. We also notice that property \ref{AAA} of Definition 
\ref{defmoreg} does not imply property \ref{BBB}.
For example, let $\mL$ be the lattice whose Hasse diagram is depicted in 
Figure \ref{fig:2}. Then $\mL$ satisfies property~\ref{AAA}, as one can easily check. However,
$\mu_\mL(S_1,T_1)=1$ and $\mu_\mL(S_2,T_1)=0$, violating 
property \ref{BBB}.

\begin{figure}[!htb]
    \centering
    \begin{minipage}{.5\textwidth}
        \centering
         \begin{tikzpicture}[scale=.6]
  \node (one) at (0,0) {$1$};
  \node (T1) at (-3,-1) {$T_1$};
    \node (T2) at (3,-1) {$T_2$};
  \node (S1) at (-3,-5) {$S_1$};
  \node (S2) at (3,-5) {$S_2$};
\node(zero) at (0,-6) {$0$};
\node(U1) at (-5,-3) {$U_1$};
\node(U2) at (-3,-3) {$U_2$};
\node(U3) at (-1,-3) {$U_3$};
\node(U4) at (1,-3) {$U_4$};
\node(U5) at (3,-3) {$U_5$};
\node(U6) at (5,-3) {$U_6$};
  \draw (one) -- (T1);
\draw  (one) -- (T2);
\draw  (T1) -- (U1);
\draw  (T1) -- (U2);
\draw  (T1) -- (U3);
\draw  (T2) -- (U4);
\draw  (T2) -- (U5);
\draw  (T2) -- (U6);
\draw  (U1) -- (S1);
\draw  (U2) -- (S1);
\draw  (U3) -- (S1);
\draw  (U4) -- (S2);
\draw  (U5) -- (S2);
\draw  (U6) -- (S2);
\draw (S1) -- (zero);
\draw (S2) -- (zero);
\end{tikzpicture}
        \caption{A regular lattice}
        \label{fig:1}
    \end{minipage}%
    \begin{minipage}{0.5\textwidth}
        \centering
         \begin{tikzpicture}[scale=.6]
  \node (one) at (0,0) {$1$};
  \node (T1) at (-3,-1) {$T_1$};
    \node (T2) at (3,-1) {$T_2$};
  \node (S1) at (-3,-5) {$S_1$};
  \node (S2) at (3,-5) {$S_2$};
\node(zero) at (0,-6) {$0$};
\node(U1) at (-5,-3) {$U_1$};
\node(U2) at (-3,-3) {$U_2$};
\node(U3) at (-1,-3) {$U_3$};
\node(U4) at (1,-3) {$U_4$};
\node(U5) at (3,-3) {$U_5$};
\node(U6) at (5,-3) {$U_6$};
  \draw (one) -- (T1);
\draw  (one) -- (T2);
\draw  (T1) -- (U1);
\draw  (T1) -- (U2);
\draw  (T1) -- (U4);
\draw  (T2) -- (U3);
\draw  (T2) -- (U5);
\draw  (T2) -- (U6);
\draw  (U1) -- (S1);
\draw  (U2) -- (S1);
\draw  (U3) -- (S1);
\draw  (U4) -- (S2);
\draw  (U5) -- (S2);
\draw  (U6) -- (S2);
\draw (S1) -- (zero);
\draw (S2) -- (zero);
\end{tikzpicture}
        \caption{A non-regular lattice}
        \label{fig:2}
    \end{minipage}
\end{figure}

We can now define the main combinatorial invariants of a regular lattice as follows.

\begin{notation}
 Let $\mL=(L,\le, \wedge, \vee)$ be a regular lattice of rank $r$.
For all integers $0 \le s,t \le r$ we define
$$\mu_\le(s,t):= |\{ S \in L : S \le T, \
\rho_\mL(S)=s \}| \ \ \ \ \ \mbox{and} \ \ \ \ \ 
\mu_\ge(s,t):= |\{ S \in L : T \le S, \ \rho_\mL(S)=s
\}|,$$
where $T \in L$ is any element with $\rho_\mL(T)=t$. For any given integers
$0 \le s \le t \le r$ we also
set   $$\mu_\mL(s,t):=\mu_\mL(S,T),$$ 
where $S,T \in L$ are arbitrary with $S \le T$, $\rho_\mL(S)=s$,
and $\rho_\mL(T)=t$. For $s>t$ we set $\mu_\mL(s,t):=0$.
\end{notation}

\begin{remark}
 A different notion of (semi)lattice regularity was proposed by
 Delsarte in \cite{del2}. The definition of Delsarte 
 is motivated by applications to Coding Theory via association schemes, rather than 
 Fourier analysis.
 Our approach and purposes are  
  different from those of \cite{del2}. For example, support maps, duality and 
  cardinality-related extremality notions are not treated in \cite{del2}.
 Notice moreover that in contrast to Delsarte's approach, in our setting the lattice structure
 is defined on an independent ``support space'', rather than on the ambient group. 
 \end{remark}

The following result easily follows from the definitions and from 
the properties of
the M\"{o}bius function. It expresses the parameters of the dual of
a regular lattice
 $\mL^*$ in terms of those of
 $\mL$.

\begin{proposition} \label{dualmre}
 Let $\mL=(L,\le,\wedge,\vee)$ be a regular lattice of rank $r$. Then
$\mL^*=(L, \preceq, \curlywedge, \curlyvee)$ is regular of rank $r$, and
for all $0 \le s,t \le r$ we have
$$\mu_\preceq(s,t)=\mu_\ge(r-s,r-t), \ \ \ \ \
\mu_\succeq(s,t)=\mu_\le(r-s,r-t), \ \ \mbox{ and } \ \ \mu_{\mL^*}(s,t)=
\mu_\mL(r-t,r-s).$$
\end{proposition}

 We conclude this section giving a sufficient condition 
 for lattice regularity that does not involve the M\"{o}bius
function. It can be used, for example, to easily test the regularity of the lattice represented in 
Figure \ref{fig:1}.

\begin{proposition}
 Let $\mL=(L,\le,\wedge,\vee)$ be a finite graded lattice.
Assume that for every $S,T \in L$ with $S \le T$ and for every 
$\rho_\mL(S) \le i \le \rho_\mL(T)$ the number
$\{ U \in L : S \le U \le T \mbox{ and } \rho_\mL(U)=i\}$
only depends on $i$, $\rho_\mL(S)$ and $\rho_\mL(T)$. Then $\mL$ is
 regular.
\end{proposition}

\begin{proof}
 Property \ref{AAA} of Definition \ref{defmoreg} is immediate, and property
\ref{BBB} can be proved by induction on $\rho_\mL(T)-\rho_\mL(S)$
using the definition of the M\"{o}bius function.
 \end{proof}
 

\section{Regular Supports and Duality} \label{secm}

In this section we propose a definition of regular support on a finite
abelian group, and establish some preliminary properties that will be needed 
in the sequel. 
In particular, we show that a regular support on a finite abelian group
$G$ canonically induces a regular support on $\hat{G}$.

As we will see in Section \ref{recover}, our definition of regular support 
generalizes both the Hamming and the rank support
for codes endowed with the Hamming and the rank weight, respectively. 
This explains the use of the word
``support'' in this paper.

\begin{notation} \label{min}
 If $G$ is a group,
$\mL=(L,\le)$ is a poset 
and $\sigma: G\to L$ is any function, then for all $S \in L$
we set
$G_\sigma(S):=\{ g \in G : \sigma(g) \le S\}$.
\end{notation}

\begin{definition}\label{defsupp0}
 Let $(G,+)$ be a
finite abelian group, and let $\mL=(L,\le, \wedge, \vee)$
be a 
regular lattice. A \textbf{regular support} 
on $G$ with values in $\mL$ is a function $\sigma:G \to L$ that satisfies the
following.
\begin{enumerate}[label=(\Alph*)]  \setlength\itemsep{0.2em}
 \item $\sigma(g)=0_{\mL}$ if and only if $g=0$. \label{uno}
 \item $\sigma(g)=\sigma(-g)$ for all $g \in G$. \label{unoprimo}
 \item $\sigma(g_1+g_2) \le \sigma(g_1) \vee \sigma(g_2)$ for all $g_1,g_2 \in
G$. \label{due}
\item $G_\sigma(S_1\vee
S_2)=G_\sigma(S_1)+G_\sigma(S_2)$ for all $S_1,S_2 \in L$. \label{tre}
\item For all $S \in L$, $|G_\sigma(S)|$ only depends on $\rho_\mL(S)$.
\label{quattro}
 \end{enumerate}
\end{definition}

Notice that properties \ref{uno}, \ref{unoprimo} and 
\ref{due} of Definition \ref{defsupp0} imply that $G_\sigma(S)$ is a subgroup of $G$
for any lattice element $S \in \mL$.

\begin{notation}
 We denote a regular support on $G$ with values in 
$\mL$ by $\sigma:G \dashrightarrow \mL$.
Moreover, for all $0 \le s \le r$ we set $$\gamma_\sigma(s):=|G_\sigma(S)|,$$
where $S \in L$ is any element with $\rho_\mL(S)=s$.
Given a lattice element $S \in L$ and a code $\mC \subseteq G$, we also define 
$\mC_{\sigma}(S):=G_\sigma(S) \cap \mC$.
\end{notation}

We can now show that the definition of regular support  behaves well
with respect to dualization. We start introducing some notation 
and establishing a preliminary lemma.

\begin{notation} \label{defstar}
 Let $\sigma:(G,+) \dashrightarrow
\mL=(L,\le,\wedge,\vee)$ be a regular support. Define
$\sigma^*:\hat{G} \to L$ by 
$$\sigma^*(\chi):= \bigvee \{ S \in L : \chi \in G_\sigma(S)^*\}$$
for all $\chi \in \hat{G}$.
Since $G_\sigma(0_\mL)=\{0\}$ by property \ref{uno} of Definition
\ref{defsupp0}, we have $\chi \in G_\sigma(0_\mL)^*$ for any
$\chi \in \hat{G}$. This shows that $\sigma^*(\chi)$ is
well defined. We regard $\sigma^*$ as a function on $\hat{G}$ with values 
in  $\mL^*$. In particular, according to Notation \ref{min},
for $S \in L$ we have
$$\hat{G}_{\sigma^*}(S)= \{ \chi \in \hat{G} : \sigma^*(\chi) \preceq S\}.$$
\end{notation}

\begin{lemma} \label{tecn}
Let $\sigma:(G,+) \dashrightarrow \mL=(L,\le,\wedge,\vee)$ be a 
 regular support. For all $\chi \in \hat{G}$ we have $\chi \in
G_\sigma(\sigma^*(\chi))^*$.
Equivalently,
$\sigma^*(\chi)$ is the maximum $S \in L$ such that $\chi \in
G_\sigma(S)^*$.
\end{lemma}

\begin{proof}
 Let $\chi \in \hat{G}$ be any character. As already shown, 
$\{ S\in L : \chi \in G(S)^*\} \neq \emptyset$. Choose an
enumeration
$\{ S\in L : \chi \in G(S)^*\}=\{ S_1,S_2,...,S_t\}$.
By property \ref{tre} of Definition \ref{defsupp0} and the
associativity of the join, we have $G_\sigma(S_1 \vee S_2 \vee \cdots
\vee S_t)=
G_\sigma(S_1)+G_\sigma(S_2)+ \cdots + G_\sigma(S_t)$.
Therefore Remark~\ref{cardual} implies
$G_\sigma(S_1 \vee S_2 \vee \cdots \vee S_t)^*=
G_\sigma(S_1)^* \cap G_\sigma(S_2)^* \cap  \cdots \cap
G_\sigma(S_t)^*$.
Since $\chi \in G_\sigma(S_i)^*$ for all $i \in \{1,...,t\}$, we have
$\chi \in G_\sigma(\sigma^*(\chi))^*$, as claimed.
\end{proof}

The following central theorem establishes the main properties of 
a regular support. In particular, it shows that
a regular support on a group $G$ with values in a lattice $\mL$
canonically induces a regular support on the character group 
$\hat{G}$,
with values in the dual lattice $\mL^*$.

\begin{theorem}  \label{fund}
Let $\sigma:(G,+) \dashrightarrow \mL =(L,\le,\wedge,\vee)$ be  a 
 regular support. The following hold.

\begin{enumerate} \setlength\itemsep{0.2em}
 \item $G_\sigma(S)^*=\hat{G}_{\sigma^*}(S)$ for all $S \in L$. \label{pa} 
 
 \item The map $\chi \mapsto \sigma^*(\chi)$ defines a
 regular support  $\sigma^*:(\hat{G},\cdot)
\dashrightarrow \mL^*=(L, \preceq,
\curlywedge,\curlyvee)$. \label{pb}

\item $\gamma_{\sigma^*}(s)=|G|/\gamma_\sigma(\mbox{rk}(\mL)-s)$ for all $0
\le s \le \mbox{rk}(\mL)$. \label{pc}

\item Identifying $\doublehat{G}$ and $G$, we have
$\sigma^{**}=\sigma$.
\end{enumerate}
\end{theorem}

\begin{proof}
 \begin{enumerate} \setlength\itemsep{0.2em}
\item Take any $S \in L$. If $\chi \in G_\sigma(S)^*$ then, by
definition, $S \le
\sigma^*(\chi)$,
i.e., $\sigma^*(\chi) \preceq S$. This shows that $G_\sigma(S)^* \subseteq
\hat{G}_{\sigma^*}(S)$.
Now assume that $\chi \in  \hat{G}_{\sigma^*}(S)$, and let $g \in G_\sigma(S)$.
We have $\sigma(g) \le S \le \sigma^*(\chi)$, hence $g \in
G_\sigma(\sigma^*(\chi))$.
Lemma \ref{tecn} implies $\chi(g)=1$, so $\hat{G}_{\sigma^*}(S)
\subseteq 
G_\sigma(S)^*$.

\item The lattice $\mL^*$ is regular
by Proposition \ref{dualmre}, and the group
$(\hat{G},\cdot)$ is finite and abelian. Let 
$\varepsilon$ be the trivial character of $G$.
By \ref{pa} we have
$\hat{G}_{\sigma^*}(0_{\mL^*})=G_{\sigma}(1_\mL)^*=G^*=\{\varepsilon\}$,
and this proves property \ref{uno} of Definition \ref{defsupp0}.
For $\chi \in \hat{G}$ and $S \in L$ we have $\chi \in G_\sigma(S)^*$ if
and only if $1/\chi \in G_\sigma(S)^*$. By definition of dual support, this
gives property \ref{unoprimo}.
Now take any 
$\chi_1, \chi_2 \in \hat{G}$, and let
$g \in G_\sigma(\sigma^*(\chi_1)) \cap G_\sigma(\sigma^*(\chi_2))$. 
Lemma \ref{tecn}  implies $\chi_1(g)=\chi_2(g)=1$, and so
$(\chi_1 \cdot \chi_2)(g)=\chi_1(g) \chi_2(g)=1$.
Therefore $$\chi_1 \cdot \chi_2 \in
(G_\sigma(\sigma^*(\chi_1)) \cap 
G_\sigma(\sigma^*(\chi_2)))^* = G_\sigma(\sigma^*(\chi_1) \wedge
\sigma^*(\chi_2))^*,$$
where the last equality directly follows from the definition of meet.
As a consequence we have $\sigma^*(\chi_1) \wedge \sigma^*(\chi_2)
\le 
\sigma^*(\chi_1 \cdot \chi_2)$, i.e., $\sigma^*(\chi_1 \cdot \chi_2) \preceq
\sigma^*(\chi_1) 
\curlyvee \sigma^*(\chi_2)$. This establishes property~\ref{due}.
Let $S_1, S_2 \in L$. By definition of meet we have 
$G_\sigma(S_1 \wedge S_2)=G_\sigma(S_1) \cap G_\sigma(S_2)$. Taking the duals,
by Remark \ref{cardual} we obtain
$G_\sigma(S_1 \wedge S_2)^*=G_\sigma(S_1)^* \cdot G_\sigma(S_2)^*$, and part 
\ref{pa} of the statement gives
$\hat{G}_{\sigma^*}(S_1 \wedge S_2)= \hat{G}_{\sigma^*}(S_1) \cdot
\hat{G}_{\sigma^*}(S_2)$,
i.e., $\hat{G}_{\sigma^*}(S_1 \curlyvee S_2)= \hat{G}_{\sigma^*}(S_1) 
\cdot \hat{G}_{\sigma^*}(S_2)$. This is property~\ref{tre}. Let $S \in L$. 
By part \ref{pa} and Remark \ref{cardual} we have
$|\hat{G}_{\sigma^*}(S)|=|G|/|G_\sigma(S)|$. Hence
$|\hat{G}_{\sigma^*}(S)|$ only depends on $\rho_{\mL^*}(S)=
\mbox{rk}(\mL)-\rho_\mL(S)$. This is property \ref{quattro}. 

\item Let $r:=\mbox{rk}(\mL)=\mbox{rk}(\mL^*)$. Take any element $S \in L$ with
$\rho_{\mL^*}(S)=s$.  Part
\ref{pa} and Remark~\ref{cardual} imply
$\hat{G}_{\sigma^*}
(S)^*=G_\sigma(S)$. Therefore 
$\gamma_{\sigma^*}(s)=|\hat{G}_{\sigma^*}(S)|=|G|/|\hat{G}_{\sigma^*}
(S)^*|=|G|/|G_\sigma(S)|=|G|/\gamma_\sigma(s)$, as desired.

\item As before, part \ref{pa} and Remark \ref{cardual} give 
 $\hat{G}_{\sigma^*}(S)^*=G_\sigma(S)$ for all $S \in L$.
 Hence, for all $g
\in G$,
$$\sigma^{**}(g) = \bigcurlyvee \{ S \in L : g \in \hat{G}_{\sigma^*}(S)^*\}
=
\bigwedge \{ S \in L : g \in G_\sigma(S)\}=\bigwedge \{ S \in L : \sigma(g) \le
S\}=\sigma(g).$$
This concludes the proof. \qedhere
\end{enumerate}
\end{proof}

\begin{definition}
 The regular support
$\sigma^*:(\hat{G},\cdot) \dashrightarrow\mL^*$ defined by 
part \ref{pb} of Theorem \ref{fund} and Notation~\ref{defstar} 
is called the \textbf{dual support} of $\sigma$.
\end{definition}

Regular supports can be constructed over any finite abelian group $G$ using as lattice any chain
of subgroups of $G$.
The regular support constructed in the following Example \ref{chsupp} will be used later to
show the existence, over any finite abelian group, of the following objects: (i) weight functions yielding MacWilliams identities,
(ii) Fourier-reflexive partitions, (iii) pairs of weights that, simultaneously, yield MacWilliams identities,
and induce metric space structures on both the underlying groups.

\begin{example}[Chain support] \label{chsupp}
 Let $(L,\le)$ be a finite chain, and let $S_0 < S_1 < \cdots < S_r$ be the
elements of $L$.
 For all $i,j \in \{0,...,r\}$ define $S_i \wedge S_j :=S_{\min\{ i,j\}}$
and $S_i \vee S_j :=S_{\max\{ i,j\}}$. Then $\mL=(L,\le,\wedge,\vee)$ is
 regular lattice of rank $r$ with:
$$\mu_\le(s,t)= \left\{ \begin{array}{rl}
                             1 & \mbox{ if $s \le t$} \\
                             0 & \mbox{ else}
                            \end{array}  \right.\  \ \ \ \ \  
   \mu_\ge(s,t)= \left\{ \begin{array}{rl}
                             1 & \mbox{ if $s \ge t$} \\
                             0 & \mbox{ else}
                            \end{array}  \right.\  \ \ \ \ \ 
  \mu_\mL(s,t)=
\left\{ \begin{array}{rl}
                             1 & \mbox{ if $s=t$} \\
                             -1 & \mbox{ if $t=s+1$} \\
                             0 & \mbox{ else}
                            \end{array}
\right.\  $$
for all $0 \le s,t \le r$.
Now let $(G,+)$ be a finite abelian group, and let 
$\mL=(L,\subseteq,\wedge,\vee)$ be a chain of subgroups of $G$, i.e.,
$\{ 0 \} = G_0 \subsetneq G_1
\subsetneq
 \cdots \subsetneq G_r=G$,  endowed with
the structure of regular lattice 
 described above. The \textbf{chain support}
$\sigma:G \dashrightarrow\mL$ is the function $\sigma:G \to L$ defined,
for all $g \in G$, by
$\sigma(g):=G_i$, where 
$i= \min \{0 \le j \le r : g \in G_j \}$. It is easy to check that 
$\sigma$ is a regular support. By definition,
$G_\sigma(G_s)=G_s$ for all $0 \le s \le r$, and therefore
$\gamma_\sigma(s)=|G_s|$ for all $s$.
Moreover,
for any $\chi \in \hat{G}$ we have
$\sigma^*(\chi)=G_i$, where $i=\max\{ 0 \le j \le r : \chi \in G_j^*\}$.
\end{example}

Notice that not all regular supports on a group $G$
arise from chains of subgroups of $G$. Other examples will be given in Section \ref{recover}
when revisiting the theory of MacWilliams identities for certain classes of codes.


\section{Compatible Weights from Regular
Supports} \label{secmw}

In this section we show that a regular support $\sigma: G \dashrightarrow \mL$ 
on a finite abelian group $G$ induces a pair of
compatible weights on $G$ and $\hat{G}$, yielding  MacWilliams identities. 
Moreover, we express the associated Krawtchouk coefficients in terms of the
combinatorial invariants of the lattice $\mL$, proving that they are 
integers with a precise combinatorial significance.
As we will see, in many relevant
examples the lattice invariants are very easy to determine. In particular, this allows to 
easily compute the Krawtchouk coefficients.
The simplification relies 
on the specific fact that the weight functions on $G$ and $\hat{G}$
both factor through a regular support map.
Whenever this happens, the following Theorem~\ref{rpm} gives an effective 
method to compute
the Krawtchouk coefficients. 

\begin{remark}
The fact that a regular support automatically yields MacWilliams identities is 
not obvious, as the definition of 
support is completely independent from the specific structure of the character group. 
This differs from previous approaches, where partitions yielding MacWilliams identities are defined (or characterized)
 in terms on their interaction with the character group. 
\end{remark}

We start observing that a regular support $\sigma: G \dashrightarrow \mL$ 
induces a weight function on 
$G$ via the rank function of $\mL$.

\begin{definition}
 Let $\sigma:(G,+) \dashrightarrow \mL$ be a regular support.
 The \textbf{weight} on $G$ \textbf{induced} by $\sigma$ is the function
$\omega_\sigma:G \to \{ 0,...,\mbox{rk}(\mL)\}$
defined by $\omega_\sigma(g):=\rho_\mL(\sigma(g))$ for all $g \in G$. 
\end{definition}

We can now state our main result.

\begin{theorem} \label{rpm}
 Let $\sigma:(G,+) \dashrightarrow \mL$ be a regular support,
$r=\mbox{rk}(\mL)$. The
 following hold.
 
 \begin{enumerate}  \setlength\itemsep{0.2em}
  \item The pair \label{111}$(\omega_{\sigma^*},\omega_\sigma)$ is compatible.
Moreover, for
  all $i \in  \omega_{\sigma^*}(\hat{G})$ and $j \in \omega_\sigma(G)$ we
have
  $$K({\omega_{\sigma^*},\omega_{\sigma}})(i,j) \ = \ \sum_{s=0}^r
\gamma_\sigma(s)
\ \mu_\mL(s,j)
 \ \mu_\le(s,r-i)  \  \mu_\ge(j,s).$$
 
 \item The pair \label{222}$(\omega_\sigma,\omega_{\sigma^*})$ is compatible.
Moreover, for
  all $i \in  \omega_\sigma(G)$ and $j \in \omega_{\sigma^*}(\hat{G})$ we have
  $$K({\omega_\sigma,\omega_{\sigma^*}})(i,j) \ = \ |G| \ \sum_{s=0}^r
\label{c2}
\frac{1}{\gamma_{\sigma}(r-s)} \
\mu_{\mL}(r-j,r-s)
 \ \mu_\ge(r-s,i)  \  \mu_\le(r-j,r-s). $$
 \end{enumerate}
 In particular, the Krawtchouk coefficients associated to both the pairs 
 $(\omega_{\sigma^*},\omega_\sigma)$ and $({\omega_\sigma,\omega_{\sigma^*}})$
 are integers determined by the combinatorial invariants of $\mL$ and $\sigma$. 
\end{theorem}

\begin{proof} Throughout this proof,
a sum over an empty set of indices is zero by definition.
 Let us first show part \ref{111}. Part \ref{222} will follow easily. Fix 
any
character 
$\chi \in \hat{G}$, and let $f,g:L \to \C$ be the 
complex-valued functions defined by
$$f(T):=\sum_{\substack{g \in G \\ \sigma(g)=T}} \chi(g), \ \ \ \ \ \ 
g(T):= \sum_{S \le T} f(S) 
 \ \ \ \  \ \ 
\mbox{ for all } T \in L.$$
By the orthogonality relations of characters (see 
\cite[Lemma 1.1.32]{vl}), for all  $T \in L$ we have
$$g(T)=\sum_{S \le T} f(S) = \sum_{g \in G_\sigma(T)} \chi(g)=
\left\{ \begin{array}{cl}
                                     \gamma_\sigma(\rho_\mL(T)) & \mbox{ 
if $\chi \in G_\sigma(T)^*$} \\
0 & \mbox{ if $\chi \notin G_\sigma(T)^*$.}
                                    \end{array}
\right.\  $$
Therefore applying the  M\"{o}bius inversion formula (e.g., 
\cite[Proposition 3.7.1]{ec}) to $f$ and $g$ we obtain

\begin{align*}
 f(T) &\ = \  \sum_{\substack{S \le T \\ \chi \in G_\sigma(S)^*}}
\gamma_\sigma(\rho_\mL(S))
\ \mu_\mL(S,T)  \ \ = \ \  \sum_{s=0}^r \sum_{\substack{ S \le T \\
\rho_\mL(S)=s \\ \chi \in
G_\sigma(S)^*}}
\gamma_\sigma(s) \ \mu_\mL(S,T) \\
&\ = \  \ \  \sum_{s=0}^r \sum_{\substack{ S \le T \\ \rho_\mL(S)=s \\ \chi \in 
\hat{G}_{\sigma^*}(S)}}
\gamma_\sigma(s) \ \mu_\mL(S,T),
\end{align*}
where the last equality follows from part \ref{pa} of Theorem \ref{fund}.
Thus for
all 
$0 \le j \le r$ one has
\begin{align}
 \sum_{\substack{g \in G \\ \omega_{\sigma}(g)=j}} \chi(g) 
&\ = \ \sum_{\substack{T \in L \\ \rho_\mL(T)=j}} f(T) \nonumber  \ \ =
\ \  \sum_{\substack{T \in L \\ \rho_\mL(T)=j}} \sum_{s=0}^r 
\sum_{\substack{ S \le T \\ \rho_\mL(S)=s  \\ \chi \in 
\hat{G}_{\sigma^*}(S)}}
\gamma_\sigma(s) \ \mu_\mL(S,T) \nonumber \\
&\ = \  \ \ \sum_{s=0}^r \gamma_\sigma(s) \sum_{\substack{T \in L \\
\rho_\mL(T)=j}}  
\sum_{\substack{ S \le T \\ \rho_\mL(S)=s \\ \chi \in 
\hat{G}_{\sigma^*}(S)}}
  \mu_\mL(S,T). \nonumber 
\end{align}
By the regularity of $\mL$,  $\mu_\mL(S,T)=\mu_\mL(s,j)$
for all $S,T \in L$ with $S \le T$, $\rho_\mL(S)=s$ and $\rho_\mL(T)=j$. 
Setting
$\alpha(s,j,\chi):=|\{ (S,T) \in L \times L :  \rho_\mL(S)=s, \ \rho_\mL(T)=j, 
\ S \le T, \ 
\sigma^*(\chi) \preceq S\}|$ we obtain
\begin{eqnarray} \label{eqprinc}
  \sum_{\substack{g \in G \\ \omega_{\sigma}(g)=j}} \chi(g) &=& 
 \sum_{s=0}^r \gamma_\sigma(s) \ \mu_\mL(s,j)
 \ \alpha(s,j,\chi).
\end{eqnarray}
We now derive a more convenient expression for $\alpha(s,j,\chi)$. By definition,
$$\alpha(s,j,\chi) = \sum_{\substack{S \in L \\ \rho_\mL(S)=s \\ \sigma^*(\chi)
\preceq S}}
|\{ T \in L : \rho_{\mL}(T)=j, \ S \le T\}| =
 \sum_{\substack{S \in L \\ \rho_\mL(S)=s \\ S \le \sigma^*(\chi) }}
\mu_\ge(j,s) \ = \ 
 \mu_\le(s,\rho_\mL(\sigma^*(\chi)))  \  \mu_\ge(j,s).$$
By the properties of the rank of the dual lattice (see Section
\ref{seclattic})
and the definition of $\omega_{\sigma^*}$, we have
$\rho_\mL(\sigma^*(\chi))=r-\rho_{\mL^*}(\sigma^*(\chi))=r-\omega_{\sigma^*}
(\chi)$. It follows 
$\mu_\le(s,\rho_\mL(\sigma^*(\chi)))  =
\mu_\le(s,r-\omega_{\sigma^*}(\chi))$, hence 
$\alpha(s,j,\chi) = \mu_\le(s,r-\omega_{\sigma^*}(\chi)) \ \mu_\ge(j,s)$.
Substituting
this expression for $\alpha(s,j,\chi)$ into~(\ref{eqprinc})
yields 
$$\sum_{\substack{g \in G \\ \omega_{\sigma}(g)=j}} \chi(g) \ = \ 
\sum_{s=0}^r \gamma_\sigma(s)
\ \mu_\mL(s,j)
 \ \mu_\le(s,r-\omega_{\sigma^*}(\chi))  \  \mu_\ge(j,s).$$
By Remark \ref{dualkk}, this shows part \ref{111}.

By Theorem
\ref{fund}, $\sigma^*$ is a regular support, and
$\sigma^{**}=\sigma$ when identifying $G$ and $\doublehat{G}$.
Therefore part~\ref{222} follows from part \ref{111} applied to
$\sigma^*:\hat{G} \dashrightarrow \mL^*$, 
along with Proposition \ref{dualmre}.

We conclude observing that the Krawtchouk coefficients are indeed integers,
as for all $0 \le s \le r$ the numbers $\gamma_\sigma(s)$ and 
$\gamma_\sigma(r-s)$ express
the cardinality of subgroups of $G$.
\end{proof}

 Let $\sigma:G \dashrightarrow \mL$ be a regular support. In the language of
  \cite{heide1},
the partitions $\mP(\omega_\sigma)$ and $\mP(\omega_{\sigma^*})$ are both
Fourier-reflexive and mutually dual, as the following result shows.
We do not go into the details of the proof.

\begin{theorem} \label{frefl}
 Let $\sigma:G \dashrightarrow \mL$ be a regular support. The partitions 
$\mP(\omega_\sigma)$ and $\mP(\omega_{\sigma^*})$ are both
Fourier-reflexive and mutually dual.
\end{theorem}

Combining Example \ref{chsupp}, Theorem \ref{rpm} and Theorem \ref{frefl} we
obtain in particular the following result of \cite{zino}, which 
shows the existence of Fourier-reflexive partitions 
on any finite abelian group. 
\begin{corollary}[Fourier-reflexive partitions via subgroups] \label{partvia}
Let $(G,+)$ be a finite abelian group, and let 
 $\{ 0 \} = G_0 \subsetneq G_1
\subsetneq
 \cdots \subsetneq G_r=G$ be a chain of subgroups of $G$.
 Then
 $$\{0\} \sqcup \bigsqcup_{i=1}^r G_i\setminus G_{i-1}$$
 is a Fourier-reflexive partition of $G$ of cardinality $r+1$.
\end{corollary}

The Fourier-reflexivity of the partitions constructed in the previous corollary was first shown
 in \cite[Theorem 6]{zino}. Notice however that our focus is  
on numerical weight functions and metric space structure associated with these 
partitions, which are not investigated in \cite{zino}. 
As we will see in the next section, the partitions of Corollary \ref{partvia} are
 induced by weight functions that endow the underlying group with a metric structure.

 This feature is relevant from a Coding Theory perspective.


\section{Metric Structures}
\label{secmetric}

 Under certain assumptions on the lattice $\mL$, the weight 
 $\omega_\sigma$ induced by a regular support 
 $\sigma:G \dashrightarrow \mL$ automatically induces 
a distance $d_{\omega_\sigma}$ on $G$. This is particularly
 interesting for applications in Coding Theory, as 
the triangle inequality is usually a key property enabling error correction.

Recall that a finite lattice $\mL=(L,\le,\wedge,\vee)$ is 
\textbf{modular}
if for all $S,T,U \in L$ with $S \le U$ one has $S \vee (T \wedge U)=(S \vee T) \wedge U$.
Notice that if $\mL$ is modular, then so is $\mL^*$. 

The following result shows that regular supports taking values 
in modular lattices induce a metric space structure on the underlying group. 

\begin{proposition} \label{metri}
 Let $\sigma:(G,+) \dashrightarrow \mL$ be a regular support.
If $\mL$ is modular, then the function $d_{\omega_\sigma}:G \times G \to \N$ defined by
$d_{\omega_\sigma}(g,g'):=\omega_\sigma(g-g')$ for all $g,g'\in G$ is a distance
function.
\end{proposition}

\begin{proof}
Write $d:=d_{\omega_\sigma}$. Let $g,g'\in G$. By definition, 
$d(g,g')=0$ if and only if $\rho_\mL(\sigma(g-g'))=0$.
By the properties of $\rho_\mL$ (Remark \ref{rk_def}), this happens if and only if
$\sigma(g-g')=0$, i.e., by property \ref{uno} of Definition \ref{defsupp0}, if and only if $g=g'$.
By property \ref{unoprimo} of Definition \ref{defsupp0} we have 
$d(g,g')=\omega_\sigma(g-g')=\rho_\mL(\sigma(g-g'))=
\rho_\mL(\sigma(g'-g))=\omega_\sigma(g'-g)=d(g',g)$.
Now let $h,g,g' \in G$.
The rank function of a modular lattice $\mL=(L,\le,\wedge,\vee)$ satisfies
$\rho_\mL(S \vee T)=\rho_\mL(S)+\rho_\mL(T)-\rho(S \wedge T)$ for all
$S,T \in L$ (see \cite{ec}, page 287). 
Thus by property \ref{due} of Definition \ref{defsupp0} we have
$$d(g,g')=\omega_\sigma(g-g')=\omega_\sigma(g-h-(g'-h)) \le 
\rho_\mL(\sigma(g-h) \vee
\sigma(g'-h)) \le d(g,h)+d(h,g').$$
This concludes the proof.
\end{proof}

\begin{remark} \label{simult}
Assume that $\sigma:(G,+) \dashrightarrow \mL$ is a regular support, with $\mL$ modular. 
Then by Theorem~\ref{fund} the support $\sigma^*:(\hat{G},\cdot) \dashrightarrow \mL^*$ 
is regular as well, where $\mL^*$ is modular. Applying Theorem \ref{rpm} and 
Proposition \ref{metri} to $\sigma$ and 
$\sigma^*$, we obtain that the weights $\omega_\sigma$ and 
$\omega_{\sigma^*}$ are bi-compatible, and induce metric space structures 
on $G$ and $\hat{G}$, respectively. This constructs a pair of metric ambient 
spaces and, \textit{simultaneously}, yields MacWilliams identities for additive codes.
\end{remark}

The following example shows that the construction 
presented in Remark \ref{simult} can be explicitly realized over any finite 
abelian group, by choosing a suitable support lattice.

\begin{example}[Chain support, continued] \label{rispp3}
Let $(G,+)$ be a finite abelian group, and let $\mL$
be a chain 
$\{ 0 \} = G_0 \subsetneq G_1
\subsetneq
 \cdots \subsetneq G_r=G$
of subgroups of $G$
endowed with the lattice structure described in Example \ref{chsupp}.
Then $\mL$ is modular. Denote by  $\sigma:G \dashrightarrow \mL$
the associated chain support. As observed in Example~\ref{chsupp},
$\sigma$ is regular. Therefore
$d_{\omega_\sigma}$ 
is a 
distance on $G$ by Proposition \ref{metri}. By Theorem \ref{fund}, $\sigma^*$ is 
a regular support. Moreover, since $\mL$ is modular, $\mL^*$ is modular.
Hence by Proposition \ref{metri} $d_{\omega_{\sigma^*}}$ 
is a 
distance on $\hat{G}$.
By Theorem \ref{rpm},
$(\omega_\sigma,\omega_{\sigma^*})$ and $(\omega_{\sigma^*},\omega_\sigma)$ are
compatible pairs
such that both $d_{\omega_\sigma}$ and $d_{\omega_{\sigma^*}}$ are distance
functions.

We conclude the example giving a more explicit description of 
$\omega_{\sigma^*}$. Let $\nu$ be
the chain support
on the character group $\hat{G}$ associated to the chain
$\{ 1\} =G_r^* \subsetneq G_{r-1}^*\subsetneq \cdots \subsetneq
\hat{G}$. 
We claim that $\omega_\nu=\omega_{\sigma^*}$ (in particular, the dual 
of a chain support is a chain support as well). Indeed, as already mentioned in Example
\ref{chsupp}, for a fixed $\chi \in \hat{G}$ 
we have
$\sigma^*(\chi)=G_i$, where $i=\max\{ 0 \le j \le r : \chi \in G_j^*\}$.
Thus, by definition, $\omega_{\sigma^*}(\chi)=\rho_{\mL^*}(\sigma^*(\chi))=r-i$.
On the other hand, $$\omega_\nu(\chi)= \min\{0 \le j \le r : \chi \in G_{r-j}^*
\}=
r-\max \{0 \le j \le r : \chi \in G_j^*\}=r-i=\omega_{\sigma^*}(\chi).$$
\end{example}

\begin{remark}
If 
$(G,+)$ is a finite abelian group, and $d: G \times G \to \R$ is a distance 
on $G$, then $d$ can be extended to a distance function $d^n$ on 
the cartesian product $G^n$ by setting
$$d^n((g_1,...,g_n),(g_1',...,g_n')):= \sum_{i=1}^n d(g_i,g_i')
\ \ \ \ \ \ \mbox{for all }  (g_1,...,g_n), (g_1',...,g_n') \in G^n.$$
It is easy to check that $d^n$ is indeed a distance function. 
Therefore a regular support $\sigma$ on $G$ taking values in a 
modular lattice $\mL$ automatically produces metric space structures on both
$G^n$ and $\hat{G}^n$. 
\end{remark}


\section{MacWilliams Identities in Coding Theory} \label{recover}

In this section we show that many weight functions traditionally 
studied in Coding Theory are induced by suitable regular supports
up to equivalence.
We also apply Theorem \ref{rpm} to easily compute the
corresponding Krawtchouk
coefficients with a unified combinatorial method. Most of such coefficients 
have been computed by other authors employing 
\textit{ad hoc} techniques in the past. Theorem~\ref{rpm}
provides a general method that applies to different
contexts. Some connections between these examples of weights 
and the general theory of group partitions 
have been studied in \cite{heide1,zino0,zino}.

Observe moreover that the case of  the rank weight (Example \ref{rwh}) is particularly
interesting, as the standard method to compute the associated
Krawtchouk coefficients is quite sophisticated~\cite{del1}. 
Theorem \ref{rpm} allows to compute them in a simple
way, and to give them a precise combinatorial interpretation.

The following Examples \ref{exe1} and \ref{rwh} also show that the 
MacWilliams identities for codes 
endowed with the Hamming and the rank weight can be seen as 
two simple instances of the same result.

\begin{example}[Additive codes with the Hamming weight] \label{exe1}
 Let $n \ge 1$ be a positive integer, and let $[n]:=\{1,...,n\}$. Then
$\mL=(2^{[n]},\subseteq, \cap,\cup)$ is a regular lattice of rank $n$. 
The rank function of $\mL$ is  the cardinality of sets. The parameters
of $\mL$ are given by
$$\mu_{\subseteq}(s,t)= \binom{t}{s}, \ \ \ \ \ \mu_{\supseteq}(s,t)=
\binom{n-t}{s-t},
\ \ \ \  \ \mu_{\mL}(s,t)= \left\{ \begin{array}{cl}
                               (-1)^{t-s} & \mbox{if } s \le t \\
0 & \mbox{if } s >t \\
                              \end{array}
 \right.\ $$
for all $0 \le s,t \le n$. The formula for $\mu_{\mL}(s,t)$ can be easily proved
by induction on $t-s$ with the aid of the Binomial Theorem (page 24 of \cite{ec}). 
See \cite[Example 3.8.3]{ec} for a different proof using the product of
chains. Let $(G,+)$ be a
finite abelian group. 
Define the \textbf{Hamming support} 
$\sigma_{\textnormal{H}}:G^n \to 2^{[n]}$
by $\sigma_{\textnormal{H}}(g_1,...,g_n):=\{ i \in [n] : g_i \neq 0\}$
for all $(g_1,...,g_n) \in G^n$. It is a regular support. The weight induced on
$G^n$ by the Hamming support
is the \textbf{Hamming weight} $\omega_{\textnormal{H}}$.
For $S
\subseteq [n]$ and $(\chi_1,...,\chi_n) \in \hat{G}^n$ we have 
$(\chi_1,...,\chi_n) \in G^n_\sigma(S)^*$  if and only if $\chi_s$
is the trivial character of $G$ for all $s \in S$.
Therefore $\sigma_{\textnormal{H}}^*(\chi_1,...,\chi_n)= \{ i \in [n] : \chi_i \mbox{ is
trivial}\}$. It follows $$\omega_{\sigma_{\textnormal{H}}^*}(\chi_1,...,\chi_n)
=n-|\{ i \in [n] : \chi_i \mbox{ is trivial}\}|=|\{ i \in [n] : \chi_i \mbox{ is
not trivial}\}|.$$
Thus in the following we write
$\omega_{\sigma_{\textnormal{H}}^*}=\omega_{\textnormal{H}}$.
 Theorem 
\ref{rpm} allows to compute the Krawtchouk coefficients for the Hamming weight
as
$$K(\omega_{\textnormal{H}},\omega_{\textnormal{H}})(i,j)=\sum_{s=0}^n
(-1)^{j-s} \ 
|G|^s \ \binom{n-i}{s} \binom{n-s}{j-s}$$
for all $0 \le i,j \le n$. By Theorem \ref{mwwide}, for every 
code $\mC \subseteq G^n$ and for all $0 \le j \le n$ we have
$$W_j(\mC^*,\omega_\textnormal{H})= \frac{1}{|\mC|} \sum_{i=0}^n
W_i(\mC,\omega_\textnormal{H}) \sum_{s=0}^n (-1)^{j-s} \ 
|G|^s \ \binom{n-i}{s} \binom{n-s}{j-s}.$$
These are the ``MacWilliams identities for the
Hamming weight over a group''.
\end{example}

\begin{example}[Linear codes with the Hamming weight]
Take $G=\F_q$ in Example \ref{exe1}. Define the \textbf{orthogonal code} of a linear
code $\mC \subseteq \F_q^n$ by $\mC^\perp:=\{ v \in \F_q^n : \langle w,v \rangle =0 
\mbox{ for all } w \in \mC \}$, where $\langle \cdot , \cdot \rangle$ is the standard inner product
of $\F_q^n$. One can show that $W_j(\mC^\perp,
\omega_\textnormal{H})=W_j(\mC^*,\omega_\textnormal{H})$
for all linear codes $\mC \subseteq \F_q^n$ (for a proof, see 
the following Example \ref{rwh}, where the same property is established for 
the more complicated case of rank-metric codes). By Example \ref{exe1}, for 
all $0 \le j \le n$ 
we have
 $$W_j(\mC^\perp,\omega_\textnormal{H})= \frac{1}{|\mC|} \sum_{i=0}^n
W_i(\mC,\omega_\textnormal{H}) \sum_{s=0}^n (-1)^{j-s} \ 
|G|^s \ \binom{n-i}{s} \binom{n-s}{j-s}.$$
These are the  ``MacWilliams identities for linear codes with the
Hamming weight''.  See for instance Chapter 5 of \cite{MS} or Chapter 7 of
\cite{pless} for equivalent formulations. 
\end{example}

\begin{example}[Modified exact weight]
Let $(G,+)$ be a non-trivial finite abelian group. Denote by $\sigma$ the chain
support on $G$ 
associated to the chain $\{0\} \subsetneq G$ (see Example \ref{chsupp}). 
Let $\omega_\sigma:G \to \{ 0,1\}$ be the induced weight.
By the second part of Example \ref{rispp3},  $\omega_{\sigma^*}$ is
the weight
on $\hat{G}$ induced by the chain support associated to the chain
$\{1\} \subsetneq \hat{G}$. 
If $n \ge 2$  and $G=\F_2$, then the $n$-th product weight of  
$\omega_\sigma$ is the \textbf{exact weight} on $\F_2^n$ 
(see \cite[page 147]{MS}).
For a general $G$, we obtain a weight that partitions the elements of 
$G^n$
according to the positions of their non-zero entries.
With the aid of  Theorem~\ref{rpm} and Example \ref{chsupp} one 
computes the 
Krawtchouk coefficients
for $(\omega_\sigma, \omega_{\sigma^*})$ and
$(\omega_{\sigma^*},\omega_\sigma)$  as
$$K(\omega_\sigma,\omega_{\sigma^*})(i,j)=
K(\omega_{\sigma^*},\omega_\sigma)(i,j)= 
\left\{\begin{array}{cl}
        1 & \mbox{ if } j=0 \\
-1 & \mbox{ if } j =1 \mbox{ and } i=1 \\
|G|-1 & \mbox{ if } j = 1 \mbox{ and } i=0 
       \end{array}
 \right.$$
for  all $i,j \in \{ 0,1\}$. 
Proposition \ref{ext}
 also allows to compute
the coefficients for
the product and symmetrized weights.
\end{example}

\begin{example}[Linear codes with the rank weight] \label{rwh}
 Let $1 \le k \le m$ be integers, and let $G:= \mbox{Mat}$ be the vector space
of $k \times m$ matrices over $\F_q$. Denote by $L$ the set of all subspaces
of
$\F_q^k$. Then $\mL=(L, \subseteq, \cap, +)$ is a regular lattice
of rank $k$. Notice that the join is  the sum of subspaces. The rank
function of $\mL$ is given by $\rho_\mL(V)=\dim(V)$ for all
 $V \subseteq \F_q^k$ (see \cite[page 281]{ec}).
The parameters of $\mL$ are, for all $0 \le s,t \le k$,
$$\mu_{\subseteq}(s,t)= {t \brack s}, \ \ \ \ \ \mu_{\supseteq}(s,t)= {k-t \brack s-t},
\ \ \ \ \ \mu_{\mL}(s,t)=\left\{ \begin{array}{cl}
                               (-1)^{t-s} q^{\binom{t-s}{2}} & \mbox{if } s \le
t \\
0 & \mbox{if } s > t, \\
                              \end{array}
 \right.\ $$
where the symbols in squared brackets are the $q$-ary binomial coefficients
(see, e.g., \cite{andr}). The formula for $\mu_{\mL}(s,t)$ can be  proved
by induction on $t-s$
with the aid of the Gaussian Binomial Theorem (\cite{ec}, equation (1.87) on page 74).
An elegant argument that uses the fact that $\mL$ is a
geometric lattice
can be found in
\cite[Example 3.10.2]{ec}.
Denote by $\mbox{colsp}(M) \subseteq \F_q^k$ the space generated by the
columns of a matrix $M \in \mbox{Mat}$.
Then $M \mapsto \mbox{colsp}(M)$ is a 
regular support 
$\sigma_{\textnormal{rk}}:\mbox{Mat} \dashrightarrow \mL$ with
$\gamma_\sigma(s)=q^{ms}$ for all $0 \le s \le k$ (see \cite[Lamma 26]{al}).
It is called the \textbf{rank support}.
Let $\omega_{\textnormal{rk}}:=\omega_{\sigma_{\textnormal{rk}}}$
be the \textbf{rank weight}, and set $\omega^*_{\textnormal{rk}}:=
\omega_{\sigma^*_{\textnormal{rk}}}$
for ease of notation. Note that 
$\omega_{\sigma_{\textnormal{rk}}}(M)=\mbox{rk}(M)$ for
all $M \in \mbox{Mat}$. By Theorem \ref{rpm},  the Krawtchouk coefficients
of
$(\omega_{\textnormal{rk}}, \omega^*_{\textnormal{rk}})$ and 
$(\omega^*_{\textnormal{rk}}, \omega_{\textnormal{rk}})$
are
\begin{equation}K({\omega_{\textnormal{rk}}, \omega^*_{\textnormal{rk}}})(i,j)= 
K({\omega^*_{\textnormal{rk}}, \omega_{\textnormal{rk}}})(i,j)
=\sum_{s=0}^k  (-1)^{j-s} \
q^{ms+\binom{j-s}{2}}  \ 
 {{k-s} \brack {k-j}} {{k-i} \brack {s}} \label{kkr}, \ \ \ \ 
0 \le i,j \le k.  \end{equation}

Recall that the \textbf{trace-product} of matrices $M,N \in \mbox{Mat}$ is
$\langle M,N \rangle := \mbox{Tr}(MN^t)$,
where $\mbox{Tr}$ is the trace of matrices, and the superscript $t$ denotes 
transposition.
The \textbf{orthogonal code} of a linear code $\mC \subseteq \mbox{Mat}$ is
$\mC^\perp:=\{ M \in \mbox{Mat} : \langle N,M\rangle =0 \mbox{ for all } N \in \mC\}$.
It can be shown that if 
$\mC \subseteq \mbox{Mat}$ is a linear code, then $W_j(\mC^\perp,\omega_{\textnormal{rk}})=
W_j(\mC^*,\omega_{\textnormal{rk}}^*)$
for all $0 \le j \le k$ (see below). Therefore combining Theorem \ref{mwwide} and equation
(\ref{kkr}) 
we obtain
$$W_j(\mC^\perp,\omega_{\textnormal{rk}})=\frac{1}{|\mC|} \sum_{i=0}^k
W_i(\mC,\omega_{\textnormal{rk}})
 \sum_{s=0}^k (-1)^{j-s} \
q^{ms+\binom{j-s}{2}}  \ 
{{k-s} \brack {k-j}} {{k-i} \brack {s}}$$
for all $0 \le j \le k$. These are the ``MacWilliams identities for linear
codes 
with the rank weight'',
first established by Delsarte in \cite{del1}.
Rank-metric codes were recently re-discovered for applications in
linear network coding (see \cite{metrics}).

We conclude this example showing that if 
$\mC \subseteq \mbox{Mat}$ is an $\F_q$-linear code, then for all 
$0 \le j \le k$ we have $W_j(\mC^\perp,\omega_{\textnormal{rk}})=
W_j(\mC^*,\omega_{\textnormal{rk}}^*)$. Fix a non-trivial 
character $\chi \in \hat{\F}_q$, and define the group 
isomorphism
$$f:\mbox{Mat} \to \widehat{\mbox{Mat}}, \ \ \ \ \ \ \ 
f(M)(N):=\chi(\mbox{Tr}(MN^t)) \ \ \mbox{for all $M,N \in \mbox{Mat}$}.$$
It is easy to see that for any linear code $\mC \subseteq \mbox{Mat}$ we have
\begin{equation} \label{sends}
f(\mC^\perp)=\mC^*.
\end{equation}

 Let 
$g:L \to L$ be the map that sends an $\F_q^k$-subspace $U$ to 
its dual $U^\perp$ with respect to the standard 
inner product of $\F_q^k$. Note that $g=g^{-1}$. For all 
$M \in \mbox{Mat}$ we have
$$\sigma_{\textnormal{rk}}^*(f(M))=\bigvee \{ S \in L : f(M) \in 
\mbox{Mat}_{\sigma_{\textnormal{rk}}}(S)^*\}=\bigvee \{ S \in L : M \in 
\mbox{Mat}_{\sigma_{\textnormal{rk}}}(S)^\perp\},$$
where the last equality follows from~(\ref{sends}) applied to
the linear code $\mbox{Mat}_{\sigma_{\textnormal{rk}}}(S)$.
By~\cite[Lemma 27]{al} and the definition of $g$ one has
$\mbox{Mat}_{\sigma_{\textnormal{rk}}}(S)^\perp = 
\mbox{Mat}_{\sigma_{\textnormal{rk}}}(S^\perp)=
\mbox{Mat}_{\sigma_{\textnormal{rk}}}(g(S))$. Therefore
$$\sigma_{\textnormal{rk}}^*(f(M)) = \bigvee \{ S \in L : M \in 
\mbox{Mat}_{\sigma_{\textnormal{rk}}}(g(S))\} = \bigvee \{ S \in L : S \subseteq 
g(\sigma_{\textnormal{rk}}(M))\}=g(\sigma_{\textnormal{rk}}(M)).$$ 
As a consequence, $g(\sigma_{\textnormal{rk}}^*(f(M)))=\sigma_{\textnormal{rk}}(M)$. Thus 
all the arrows in the following 
diagram commute, showing that the $\omega_{\textnormal{rk}}$-distribution of 
$\mC^\perp$ coincides with the $\omega_{\textnormal{rk}}^*$-distribution of $\mC^*$.
$$\xymatrix{
\mbox{Mat} \ar[rr]^{f} \ar[d]_{\sigma_{\textnormal{rk}}} &  
 & \widehat{\mbox{Mat}} \ar[d]^{\sigma_{\textnormal{rk}}^*}\\
L   \ar[rd]_{\rho_\mL}& & L  \ar[ll]^{g} \ar[ld]^{\rho_{\mL^*}} \\
& \N & 
}$$
\end{example}

\begin{example}[Lee weight on $\Z_4$]
The \textbf{Lee weight} on $\Z_4$ is the function
$\omega_{\textnormal{Lee}}:\Z_4 \to \{ 0,1,2\} \subseteq \N$ defined by
$\omega_{\textnormal{Lee}}(0):=0$, $\omega_{\textnormal{Lee}}(1)=
\omega_{\textnormal{Lee}}(3):=1$
and $\omega_{\textnormal{Lee}}(2):=2$. See \cite{z4} and \cite{leew} or 
Chapter 12 of \cite{pless} and the references within.
Denote by $\sigma$ the chain support on $\Z_4$ associated to the chain
$\{ 0\} \subsetneq \Z_2 \subsetneq \Z_4$. Then 
$\omega_\textnormal{Lee} \sim \omega_\sigma$.
Let $\zeta \in \C$ be a primitive fourth root of unity. Define the map 
$\psi:\Z_4 \to \hat{\Z}_4$ by
$\psi(a)(b):=\zeta^{ab}$ for all $a,b \in \Z_4$. Then $\psi$ is a group 
isomorphism, and it is natural to define the \textbf{Lee weight} on
$\hat{\Z}_4$ by $\omega_\textnormal{Lee}^*:=
\omega_\textnormal{Lee} \circ \psi^{-1}$. A direct computation shows
$\omega_{\sigma}=\omega_{\sigma^*} \circ \psi$, and therefore
$\omega_\textnormal{Lee}^*=
\omega_\textnormal{Lee} \circ \psi^{-1} \sim \omega_\sigma \circ \psi^{-1}=
\omega_{\sigma^*} \circ \psi \circ \psi^{-1}=\omega_{\sigma^*}$. Thus the 
Krawtchouk coefficients associated to 
$(\omega_\textnormal{Lee},\omega_\textnormal{Lee}^*)$
are the same as the Krawtchouk coefficients associated to 
$(\omega_\sigma,\omega_{\sigma^*})$,
up to a permutation. They can be explicitly computed combining 
Example \ref{chsupp} and Theorem \ref{rpm} as follows. We write
$K_\textnormal{Lee}$ for $K(\omega_\textnormal{Lee},\omega_\textnormal{Lee}^*)$.
\begin{align*}
 K_\textnormal{Lee}(0,0) &= 1 &  
K_\textnormal{Lee}(0,1) &= 2 &   
K_\textnormal{Lee}(0,2) &= 1 \\
K_\textnormal{Lee}(1,0) &= 1 &  
K_\textnormal{Lee}(1,1) &= 0 &   
K_\textnormal{Lee}(1,2) &= -1 \\
K_\textnormal{Lee}(2,0) &= 1 &  
K_\textnormal{Lee}(2,1) &= -2 &   
K_\textnormal{Lee}(2,2) &= 1.
\end{align*}
Proposition \ref{ext} also allows to compute the
Krawtchouk coefficients for the 
\textbf{symmetrized Lee weight} on the product group $\Z_4^n$, for $n \ge 1$ (see e.g. \cite{z4}).
\end{example}

\begin{example}[Homogeneous weight on certain Frobenius rings] \label{hhhhh}
 We denote the socle and the Jacobson radical of a finite (possibly non-commutative) 
Frobenius
ring $R$ by $\mbox{soc}(R)$ and $\mbox{rad}(R)$, respectively.
See \cite[Chapter 16]{lam}
for the main properties of Frobenius rings, or \cite{gref} and
\cite{heide2} for a 
Coding Theory approach. 
It is known that  $\mbox{rad}(R)$ is a two-sided ideal, and that 
$\mbox{soc}(R)
\cong R/\mbox{rad}(R)$ as left and right $R$-modules. Moreover,
if $R$ is local, i.e., $\mbox{rad}(R)$ is the unique maximal left and right
ideal of $R$, then $R/\mbox{rad}(R)$ is a field, called the \textbf{residue
field}.

Let $R:=R_1 \times R_2 \times \cdots \times R_n$, where each $R_i$ is 
a finite local Frobenius ring. Then  
 $R_i/\mbox{rad}(R_i) \cong \mbox{soc}(R_i)$
 as left and right $R_i$-modules. We assume that all the residue
 fields
 $R_i/\mbox{rad}(R_i)$ have the same 
 order $q$. Then $R$ is Frobenius with 
$\mbox{soc}(R)= \prod_{i=1}^n \mbox{soc}(R_i)$.
The values of the \textbf{homogeneous weight} $\omega_\textnormal{hom}:R
\to \R$ (see \cite{cost,gref,hon}) 
on $R$ were explicitly computed in
\cite[Proposition 3.8]{heide2} as
$$\omega_{\textnormal{hom}}(a)= \left\{ \begin{array}{cl}
 1-{\left( \frac{-1}{q-1}\right)}^{\textnormal{wt}(a)} & \mbox{if } a \in \mbox{soc}(R) \\
1 & \mbox{otherwise,} \end{array}
\right.\ $$
where $\mbox{wt}(a):= |\{ 1 \le i \le n : a_i \neq 0 \}|$ is the  weight of
$a=(a_1,...,a_n)$.

From now on we assume $q \ge 3$. In particular,
we have $\omega_{\textnormal{hom}}(a)=0$
if and only if $a=0$.
Let $[n+1]:=\{ 1,...,n+1\}$, and 
$L:= \{ S \subseteq [n+1] : n+1 \notin S\} \cup \{ [n+1]\}$.
Then $\mL=(L,\subseteq , \cap,\cup)$ is a regular lattice of rank $n+1$, where
the rank 
function is
given by the cardinality of sets. It is easy to see that the parameters of $\mL$ are,
for all $0 \le s,t \le n+1$,
\begin{equation*}\mu_{\subseteq}(s,t)= \left\{ \begin{array}{cl}
\binom{t}{s} & \mbox{ if } s \le t \le n \\ 
\binom{n}{s} & \mbox{ if } s \le n, t=n+1 \\ 
1 & \mbox{ if } s=t=n+1 \\
0 & \mbox{ if } s>t,\end{array}
\right.\ 
 \ \ \ \ \ \ \ \ 
\mu_{\supseteq}(s,t)= \left\{ \begin{array}{cl}
\binom{n-t}{s-t} & \mbox{ if } t \le s \le n \\ 
1 & \mbox{ if } t \le s=n+1 \\ 
0 & \mbox{ if } s<t,\end{array}
\right.\ 
\end{equation*}

\begin{equation*}
\mu_\mL(s,t)= \left\{ \begin{array}{cl}
                       {(-1)^{t-s}} & \mbox{if } s \le t \le n \\
0 & \mbox{if } t<s, \mbox{ or } t=n+1 \mbox{ and } s < n  \\
-1 & \mbox{if } t=n+1, s=n.
                      \end{array}
\right.\  
\end{equation*}
The formula for $\mu_\mL(s,t)$ can be proved by induction on $t-s$
using the Binomial Theorem, as in Example \ref{exe1}.
Define $\sigma:R \to L$ by
$\sigma(a):=[n+1]$ if $a \notin \mbox{soc}(R)$, and 
$\sigma(a):=\{ 1 \le i \le n : a_i \neq 0\}$ if $a \in \mbox{soc}(R)$. One can
check that
$\sigma:R \dashrightarrow \mL$
is a regular support with
$$\gamma_\sigma(s)= \left\{ \begin{array}{cl}
                             q^s & \mbox{if } s \le n \\
|R| & \mbox{if } s=n+1
                            \end{array}
\right.\ $$
for all $0 \le s \le n+1$. 
Moreover, $\omega_\sigma \sim \omega_{\textnormal{hom}}$. 
 By Theorem \ref{frefl}, in the language  
 of~\cite{heide1}
we have
$$\widehat{\mP(\omega_{\textnormal{hom}})}=\mP(\omega_{\sigma^*}).$$
Therefore the Krawtchouk matrix ${\bf{K}}$ associated to the homogeneous weight
partition
(see Section~4 of~\cite{heide2} for the definition) is given by
\begin{equation} \label{kmat}
 {\bf{K}}_{ij}:= K(\omega_{\sigma^*},\omega_\sigma)(i,j)
\end{equation}
for all $0 \le i,j \le n+1$. 

When $n=1$, the ring $R=R_1$ is a finite local Frobenius ring, and with 
the aid of Theorem~\ref{rpm} one can easily compute 
$${\bf{K}}= \begin{bmatrix}
             1 & q-1 & |R|-q \\ 1 & q-1 & -q \\ 1 & -1 & 0
            \end{bmatrix}.
$$
The same matrix appears in \cite{camion} and \cite{heide2}  
for $R=\Z_8$. 

Combining equation (\ref{kmat}) and Theorem \ref{rpm},
one obtains new explicit formul{\ae}  for the Krawtchouk 
coefficients associated to the homogeneous weight, along with a 
combinatorial interpretation for them. 
Since $\mL$ is modular, 
by
Proposition \ref{metri} the weight function $\omega_\sigma$
automatically induces a distance function on $R$.

Note that for some simple Frobenius rings it is possible to express the homogenous weight via a suitable
chain support on the ring. For example, the homogeneous weight on a finite local Frobenius ring
$R$ is equivalent to the chain support associated to
the chain $0 \subsetneq \mbox{soc}(R) \subsetneq R$
 (see \cite{eb} or \cite{heide2} for the values of the homogeneous weight on
such rings).

It is known (see \cite[Section 4]{heide2}) that the partition induced by 
the homogeneous weight on more general Frobenius rings is not 
Fourier-reflexive. In particular, there is no regular support defined on 
these rings that induces the homogeneous weight. This is the case, for 
example, of the ring 
$\Z/546\Z$ (see \cite[Example 4.6]{heide2} for details).
\end{example}


\section{Extremality} \label{secop}

In this section we study subsets $\mC \subseteq G$ that are not necessarily subgroups of $G$.
We consider a slightly more general setting than the one we investigated in the previous sections, 
relaxing the definition of regular support (see the following Notation \ref{not2}).
We establish a generalized Singleton bound for subsets $\mC \subseteq G$. We 
call ``extremal'' the codes attaining the Singleton bound. This yields a cardinality-related 
notion of extremality for codes in groups, that 
extend the concept of MDS code and MRD rank-metric code. 
Our specific interest in codes' cardinality is motivated by the fact that 
this fundamental parameter reflects the code's rate.

We show that if $\mC$ is an extremal set, then the weight distribution of any translate of $\mC$ and 
the distance distribution of $\mC$ coincide. Moreover, they can be expressed in terms 
of the combinatorial invariants
of the underlying lattice. 
Finally, we prove that if $\mC$ is an extremal subgroup
(i.e., an extremal code), then the dual code $\mC^*$ is extremal as well.

\begin{remark}
Our approach extends classical results on the distance distributions of extremal codes 
to the weight distributions of their translates. Notice that full knowledge of the weight 
distributions of the translates of a code gives information on the distance distribution 
of the code itself. However, the converse is not true in general, and the weight 
distributions of the translates of a code need  an independent analysis. 
For MDS codes endowed with the Hamming weight, this analysis is carried 
out e.g. in \cite{B90} with an elegant simple argument.
\end{remark}

\begin{notation} \label{not2}
In this section $(G,+)$ is a finite abelian group, and 
 $\mL=(L,\le,
\wedge,\vee)$ denotes a finite graded lattice of rank $r$ that satisfies property
\ref{AAA} of Definition \ref{defmoreg}. Moreover,
$\sigma:G \to L$ is a function that satisfies properties 
\ref{uno}, \ref{unoprimo}, \ref{due} and \ref{quattro} 
of Definition \ref{defsupp0}. We simply denote by $\omega:G\to \{0,...,r\}$ the function 
defined by $\omega(g):=\rho_{\mL}(\sigma(g))$ for all $g \in G$.
We follow the notation of the previous sections, unless specified differently. In particular, we set 
$\mC_\sigma(S):=\{g \in \mC : \sigma(g) \le S\}$
for any (possibly non-additive) subset $\mC \subseteq G$ and $S \in L$.
\end{notation}

In the sequel we investigate combinatorial properties of subsets 
$\mC \subseteq G$ that are not necessarily subgroups of $G$.

\begin{definition}
 Let $\mC \subseteq G$ be any subset with $|\mC| \ge 2$.
 The \textbf{minimum weight} and the 
 \textbf{minimum distance}
 of $\mC$ are, respectively, 
 $$w_\omega(\mC):=\min \{ \omega(g) : g \in \mC, \ g \neq 0 \}, 
 \ \ \ \ \ d_\omega(\mC):=\min \{ \omega(g-g') : g,g' \in \mC, \ g \neq g' \}.$$
 The \textbf{weight} and \textbf{distance distributions} of $\mC$ are
  the integer vectors 
  $(W_i(\mC,\omega) : i=0,...,r)$ and $(D_i(\mC,\omega) : i=0,...,r)$, respectively, 
  where $$W_i(\mC,\omega):=|\{ g \in \mC : \omega(g)=i\}|, \ \ \ \ \ 
  D_i(\mC,\omega):=\frac{1}{|\mC|} \ |\{ (g,g') \in \mC^2 : \omega(g-g')=i\}|$$
for all $i \in \{0,...,r\}$.
\end{definition}

Notice that we do not require the map $G \times G \to \{0,...,r\}$ given by $(g,g')\mapsto \omega(g-g')$
to be a distance function on $G$.

\begin{remark}
 It is easy to check that if $\mC \subseteq G$ is a subgroup (i.e., a code) then 
  $W_i(\mC,\omega)=D_i(\mC,\omega)$
 for all $i=0,...,r$. In particular, if $|\mC| \ge 2$ then $w_\omega(\mC)=d_\omega(\mC)$.
\end{remark}

We start with a Singleton-type bound of combinatorial flavor.

\begin{proposition}\label{LAsbound}
 Let $\mC \subseteq G$ be a subset with $|\mC| \ge 2$. We
have $|\mC| \le |G|/\gamma_\sigma(d_\omega(\mC)-1)$.
\end{proposition}

\begin{proof}
 Take any $S \in L$ with $\rho_\mL(S)=d_\omega(\mC)-1$. Such an $S$ always exists
 by definition of rank of a graded poset. For all $g \in \mC$ define  
 $$[g]:=g+G_\sigma(S)=\{g+h : h \in G_\sigma(S)\}\subseteq G.$$
 By definition of minimum distance we have $[g] \cap [g']=\emptyset$
 for all $g,g' \in \mC$ with $g \neq g'$. Therefore 
 \begin{equation*}|G| \ge \left| \bigcup_{g \in \mC} [g] \right|= \sum_{g \in \mC} |[g]|
 = \sum_{g \in \mC} |G_\sigma(S)| = |\mC| \cdot \gamma_\sigma(d_\omega(\mC)-1),
 \end{equation*}
 and the bound follows.
 \qedhere
 \end{proof}

 \begin{definition}
  A subset $\mC \subseteq G$ is  \textbf{extremal} if
  $|\mC| \ge 2$ and
it attains the bound of Proposition \ref{LAsbound}. 
 \end{definition}

The remainder of the section is devoted to the combinatorial properties of extremal sets.
We start with a preliminary result.

 \begin{lemma}\label{LAprell}
  Let $\mC \subseteq G$ be an extremal subset. Let 
  $S \in L$ be any element with $s:=\rho_\mL(S) \ge d_\omega(\mC)$. Then 
  $$|\mC_\sigma(S)|=\frac{|\mC| \ \gamma_\sigma(s)}{|G|}.$$
   \end{lemma}

 \begin{proof}
  Let
  $T \in L$ with $T \le S$ and $\rho_\mL(T)=d_\omega(\mC)-1$.
  Such a $T$ always exists by definition of graded posets.
  We clearly have $G_\sigma(T) \subseteq G_\sigma(S)$.
  Define the maps 
  $$\mC \stackrel{\pi_1}{\longrightarrow} G/G_\sigma(T) 
  \stackrel{\pi_2}{\longrightarrow}
  G/G_\sigma(S)$$
  as follows.
  The function $\pi_1$ is the composition of the inclusion $\mC \to G$ and the projection 
  on the quotient group $G \to G/G_\sigma(T)$.
  The map $\pi_2$ is given by $g+G_\sigma(T) \mapsto g+G_\sigma(S)$, and it is a well defined 
  group homomorphism, as $G_\sigma(T) \subseteq G_\sigma(S)$.

  We claim that $\pi_1$ is a bijection. Indeed, assume that there exist 
  $g,g' \in \mC$ with $\pi_1(g)=\pi_1(g')$, i.e., $g+G_\sigma(T)=g'+G_\sigma(T)$.
  Then $g-g' \in G_\sigma(T)$, hence $\omega(g-g') \le \rho_\mL(T)=d_\omega(\mC)-1$.
  It follows $g=g'$, i.e., $\pi_1$ is injective. Since $\mC$ is extremal, we have 
  $|\mC|=|G|/\gamma_\sigma(d_\omega(\mC)-1)=|G|/G_\sigma(T)$, and so $\pi_1$ is a bijection,
  as claimed.
  
  Since both $\pi_1$ and $\pi_2$ are surjective, the map $\pi:=\pi_2 \circ \pi_1$ 
   is surjective as well. Moreover, as $\pi_1$ is bijective and 
   $\pi_2$ is a surjective group homomorphism,
   we have 
   $|\pi^{-1}(0)| = |\pi^{-1}(x)|$ for all $x \in G/G_\sigma(S)$. Therefore
   $$|\mC|= \left| \bigcup_{x \in G/G_\sigma(S)} \pi^{-1}(x) \right|=
   \sum_{x \in G/G_\sigma(S)} |\pi^{-1}(x)| =
   \sum_{x \in G/G_\sigma(S)} |\pi^{-1}(0)| = \frac{|G|}{\gamma_\sigma(s)} \cdot |\mC_\sigma(S)|,$$
   where the last equality follows from the definition of $\mC_\sigma(S)$.
  This shows the lemma.
   \end{proof}

 \begin{theorem} \label{LAfff}
  Let $\mC \subseteq G$ be an extremal subset of minimum distance 
  $d:=d_\omega(\mC)$ and $0 \in \mC$. 
  Define the integer matrix $P$ of size $(r-d+1) \times (r-d+1)$ by
  $P_{ij}:=\mu_\ge(d+i-1,d+j-1)$ for all 
  $i,j \in \{1,...,r-d+1\}$. Then $P$ is invertible, and the weight distribution of 
  $\mC$ is given by 
  $$
     W_0(\mC,\omega)=1, \ \ W_i(\mC,\omega)=0 \mbox{ for } 1 \le i \le d-1, $$ $$ 
     \left( \begin{array}{c}
     W_d(\mC,\omega) \\ W_{d+1}(\mC,\omega) \\ \vdots \\ W_r(\mC,\omega)\end{array} \right)=
     P^{-1} \left( \begin{array}{c}
     |\mC| \ \mu_\le(d,r) \gamma_\sigma(d)/ |G| -\mu_\ge(d,0)\\ |\mC| \ 
     \mu_\le(d+1,r) \gamma_\sigma(d+1)/ |G| - \mu_\ge(d+1,0)\\ \vdots \\
     |\mC| \ \mu_\le(r,r) \gamma_\sigma(r)/ |G| - \mu_\ge(r,0)\end{array} \right).    
$$
  In particular,  the weight distribution of $\mC$ only depends on $|G|$, $d_\omega(\mC)$,
  and on the combinatorial invariants of $\mL$ and $\sigma$.
 \end{theorem}
 
 \begin{proof}
  Take any $s \in \{d,...,r\}$ and write $d:=d_\omega(\mC)$. 
  We will count the elements of the set 
  $$\mA:=\{(g,S) : g \in \mC, \ S \in L, \ \rho_\mL(S)=s, \ \sigma(g) \le S\}$$
  in two different ways. On the one hand, by Lemma \ref{LAprell} we have 
 $$|\mA|=\sum_{\substack{S \in L \\ \rho_\mL(S)=s}} |\mC_\sigma(S)|=\mu_\le(s,r) \ \frac{|\mC| \ 
 \gamma_\sigma(s)}{|G|}.$$
 Since $0 \in \mC$, by definition of 
 $\mA$ we have
 $$|\mA|=\sum_{i=0}^s \sum_{\substack{g \in \mC \\ \omega(g)=i}} |\{S \in L :  
 \rho_\mL(S)=s, \ \sigma(g) \le S\}|=
\mu_\ge(s,0) + \sum_{i=d}^s W_i(\mC,\omega) \ \mu_{\ge}(s,i).$$
 Therefore 
 \begin{equation} \label{xxx}
 \sum_{i=d}^s W_i(\mC,\sigma) \ \mu_\ge(s,i) \ = \ \mu_\le(s,r) \ \frac{|\mC| \ 
 \gamma_\sigma(s)}{|G|} - \mu_\ge(s,0),
 \ \ \ \ \ s \in \{d,...,r\}.
 \end{equation}
 Observe that (\ref{xxx}) is a system of $r-d+1$ linear equations in the 
 unknowns $W_d(\mC,\omega),...,W_r(\mC,\omega)$. The matrix of the system is precisely 
 $P$, which is lower triangular with all ones on the diagonal.
 \end{proof}

 \begin{remark}
 Following the notation of 
 Theorem \ref{LAfff}, when $\mL$ also satisfies property \ref{BBB} of Definition~\ref{defmoreg}
 the weight distribution of an extremal subset $\mC \subseteq G$  with $d=d_\omega(\mC)$ and $0 \in \mC$ can be expressed as
 $$W_i(\mC,\sigma)= \sum_{s=0}^{d-1} \mu_\mL(s,i) \ \mu_\le(s,i) + 
 \sum_{s=d}^i \mu_\mL(s,i) \ \mu_\le(s,i) \ \frac{|\mC| \ \gamma_\sigma(s)}{|G|},\ \ \ \  d \le i \le r.$$
  We do not go into the details of the proof.
  In the context of codes endowed with the rank metric, to our best 
  knowledge the previous formula
is new. As we will see in  Corollary \ref{LAccc}, it generalizes a result by Delsarte 
on the distance distribution of extremal codes (see \cite[Theorem 5.6]{del1}).
 \end{remark}

 If $\mC \subseteq G$ is a non-empty subset and $h\in G$, we define the \textbf{translate} of $\mC$ by 
 $h$ as  the set $\mC_h:=\{g-h : g \in \mC\} \subseteq G$.

 \begin{corollary} \label{LAccc}
  Let $\mC \subseteq G$ be an extremal subset of minimum distance 
  $d:=d_\omega(\mC)$. For all $h \in \mC$ we have 
  $W_i(\mC_h,\omega)=D_i(\mC,\omega)$ for $i \in \{0,...,r\}$.
 \end{corollary}

 \begin{proof}
   For all  $i \in \{0,...,r\}$ one has
  $$D_i(\mC,\omega) = \frac{1}{|\mC|} \ |\{ (g,g') \in \mC^2 :
\omega(g-g')=i\}| =
  \frac{1}{|\mC|} \sum_{g' \in \mC} |\{ g \in \mC : \omega(g-g')=i\}| =
    \frac{1}{|\mC|} \sum_{g' \in \mC} W_i(\mC_{g'},\omega).$$
    Let $h \in \mC$ be any element. It is easy to see that the
translate $\mC_h$ has the
same distance distribution as 
  $\mC$. In particular, $\mC_h$ is extremal. Moreover, since $0 \in \mC_h$,
its weight distribution is given by 
  Theorem~\ref{LAfff}, and it does not depend on $h \in \mC$.
  Therefore for all $h \in \mC$ and  $i \in \{0,...,r\}$ one has  
  $$D_i(\mC,\omega)= \frac{1}{|\mC|} \cdot |\mC| \cdot  W_i(\mC_h,\omega) =
W_i(\mC_h,\omega),$$
as claimed. 
 \end{proof}

 \begin{remark}
Combining Theorem \ref{LAfff} and Corollary \ref{LAccc} we obtain generalizations 
of classical results on the distance distribution of the translates of
cardinality-optimal  
codes to the general context of additive codes in groups.
 \end{remark}

 We conclude this section showing that if $\mC$ is an extremal code 
 (i.e., a subgroup of $G$) and $\sigma$ 
 also satisfies property \ref{tre} of 
 Definition \ref{defsupp0}, then the dual of $\mC$ is an extremal code as well. 
 Note that property \ref{BBB} of Definition \ref{defmoreg}
 is still not required, and thus $\mL$ is not a regular lattice in general.

\begin{lemma} \label{prelmww1}
Let $\mC \subseteq G$ be a code, and assume that 
$\sigma$ also satisfies property \ref{tre} of 
 Definition \ref{defsupp0}.
Take any $S \in L$, and let $s:=\rho_\mL(S)$. We have
$$\displaystyle |\mC_\sigma(S)|=\frac{|\mC| \cdot
|\mC^*_{\sigma^*}(S)|}{\gamma_{\sigma^*}(r-s)}.$$
\end{lemma}

\begin{proof}
By definition, $\mC_\sigma(S) = G_\sigma(S) \cap \mC$. Remark \ref{cardual}
implies 
\begin{equation} \label{pppi}
 |\mC_\sigma(S)|=\frac{|G_\sigma(S)| \cdot |\mC|}{|G_\sigma(S)+\mC|}=
 \frac{|G_\sigma(S)| \cdot |\mC| \cdot 
|(G_\sigma(S)+\mC)^*|}{|G|}.
\end{equation}
Again by Remark \ref{cardual} we have 
$|(G_\sigma(S)+\mC)^*|=|G_\sigma(S)^* \cap \mC^*|=|\hat{G}_{\sigma^*}(S) \cap
\mC^*|$, where the last equality follows from Theorem \ref{fund}
(whose proof does not require property \ref{BBB} of Definition \ref{defmoreg}). 
Since $\hat{G}_{\sigma^*}(S) \cap
\mC^* =\mC^*_{\sigma^*}(S)$ by definition, equation (\ref{pppi}) can be
written as
\begin{equation*}
 |\mC_\sigma(S)|=\frac{|G_\sigma(S)|\cdot |\mC| \cdot |\mC^*_{\sigma^*}(S)|}{|G|}.
\end{equation*}
By Theorem \ref{fund} we have $|G|/|G_\sigma(S)|=
|G|/\gamma_\sigma(s)=\gamma_{\sigma^*}(r-s)$,
and the result follows.
\end{proof}

\begin{theorem} \label{dualextremal}
 Let $\mC \subseteq G$ be a non-trivial extremal code, and assume that 
$\sigma$ also satisfies property~\ref{tre} of 
 Definition~\ref{defsupp0}.
Then $d_{\omega_{\sigma^*}}(\mC^*) \ge r-d_{\omega_{\sigma}}(\mC)+2$, and 
the code $\mC^*$ is extremal.
\end{theorem}

\begin{proof}
Let $d:=d_{\omega_\sigma}(\mC)$ and $d^*:=d_{\omega_{\sigma^*}}(\mC^*)$. 
Since $\mC$ is  extremal, we have $|\mC| = |G|/\gamma_\sigma(d-1)$. 
Remark~\ref{cardual} and Theorem \ref{fund} (whose proof does not require 
property \ref{BBB} of Definition \ref{defmoreg}) imply
\begin{equation} \label{ee}
 |\mC^*|=\gamma_\sigma(d-1)=|\hat{G}|/\gamma_{\sigma^*}(r-d+1).
\end{equation}
Let $S \in L$ be any element with $\rho_{\mL^*}(S) = r-d+1$. 
Then
$\rho_\mL(S)=r-(r-d+1)=d-1$, and so
$\mC_\sigma(S)=\{0\}$. Lemma \ref{prelmww1} gives
$$|\mC^*_{\sigma^*}(S)| = \frac{|\mC_\sigma(S)| \cdot
\gamma_{\sigma^*}(r-d+1)}{|\mC|}
= 1,$$ where the last equality easily follows from equation (\ref{ee}) and
Remark \ref{cardual}.
Therefore $\mC^*_{\sigma^*}(S)=\{0\}$, and the minimum weight/distance
of $\mC^*$
 satisfies
$d^* \ge r-d+2$. In particular, 
$\gamma_{\sigma^*}(d^*-1) \ge \gamma_{\sigma^*}(r-d+1)$.
Combining Proposition \ref{LAsbound} applied to $\mC^*$ and $\sigma^*$
with equation (\ref{ee}) we obtain
\begin{equation*} \label{long}
 \frac{|\hat{G}|}{\gamma_{\sigma^*}(r-d+1)} \ge 
\frac{|\hat{G}|}{\gamma_{\sigma^*}(d^*-1)} \ge 
|\mC^*|=\frac{|\hat{G}|}{\gamma_{\sigma^*}(r-d+1)}.
\end{equation*}
It follows $|\mC^*|=|\hat{G}|/ \gamma_{\sigma^*}(d^*-1)$, i.e.,
$\mC^*$ is extremal.
\end{proof}


\section{Enumerative Problems of Matrices} \label{enmat}

In this section we  show how one can apply
the MacWilliams identities
for the rank weight 
to answer some open enumerative combinatorics questions
on matrices over a finite field. In particular, we answer a 
generalized question of R. Stanley on the number of matrices 
with given rank and zero diagonal entries.

Following the notation of Example \ref{rwh}, in the sequel $k$ and
$m$ are integers with $1 \le k \le m$, 
and $\F_q$ is the finite field with $q$ elements. We denote by $\mbox{Mat}$ the
$km$-dimensional
 space of $k \times m$ matrices over $\F_q$. Given an integer $s \ge 1$, 
we set $[s]:=\{1,...,s\}$. 
The rank support
on $\mbox{Mat}$ is denoted by $\sigma_{\textnormal{rk}}$, 
and $\omega_{\textnormal{rk}}$ is the rank weight. We write 
``rank-distribution'' for ``$\omega_{\textnormal{rk}}$-distribution''. 
 All dimensions in this section 
are computed 
over $\F_q$. 

Recall 
 that the \textbf{trace-product} of 
matrices $M,N \in \mbox{Mat}$ is  
$\langle M,N \rangle := \mbox{Tr}(MN^t)$, where $\mbox{Tr}$ denotes the 
trace of matrices and $t$ denotes transposition. The \textbf{orthogonal} of a
linear code
$\mC \subseteq \mbox{Mat}$ is the linear code
$\mC^\perp=\{ M \in \mbox{Mat} : \langle M,N \rangle =0 \mbox{ for all } N \in
\mC\}$. Notice that for any linear code $\mC$ one has 
$\{M^t : M \in \mC\}^\perp = \{M^t : M \in \mC^\perp\}$.
In particular, up to a transposition of the matrices, the assumption
 $k \le m$ is not restrictive.
 
We start by recalling the 
MacWilliams identities for linear codes endowed with the rank weight 
(see  \cite{del1} or Example \ref{rwh}). 

\begin{theorem} \label{recallmw}
Let $\mC \subseteq \mbox{Mat}$ be a linear code.
The rank-distributions of $\mC$ and $\mC^\perp$ satisfy
 $$W_j(\mC^\perp,\omega_{\textnormal{rk}})=\frac{1}{|\mC|} \sum_{i=0}^k
W_i(\mC,\omega_{\textnormal{rk}})
 \sum_{s=0}^k (-1)^{j-s} \
q^{ms+\binom{j-s}{2}}  \ 
{{k-s} \brack {k-j}} {{k-i} \brack {s}}$$
for all $0 \le j \le k$.
In particular, they determine each other.
\end{theorem}

The first enumerative technique that we present is based on the following simple observation. If 
$f:\mbox{Mat} \to \F_q$ is a non-zero
$\F_q$-linear function,
then $\mbox{ker}(f)^\perp$ is a linear code  generated by one matrix.
Any two generating matrices have the same rank, say
$R_f$.
Thus the rank distribution of the linear code $\mC:=\mbox{ker}(f)^\perp$ is
$$W_i(\mC,\omega_{\textnormal{rk}})= \left\{ \begin{array}{cl}
                              1 & \mbox{ if $i=0$} \\
q-1 & \mbox{ if $i=R_f$} \\
0 & \mbox{ otherwise}.
                             \end{array}
\right.\ $$
Applying Theorem \ref{recallmw} to $\mC:= \mbox{ker}(f)^\perp$
one can now explicitly compute the number of matrices of rank
$j$ in $\mbox{ker}(f)=\mC^\perp$ for all $0 \le j \le k$.
More precisely, the following hold.

\begin{corollary} \label{kerf}
 Let $f:\mbox{Mat} \to \F_q$ be a non-zero linear map, and let
$R_f$ be the rank of any matrix that generates $\mbox{ker}(f)^\perp$.
For all $0 \le j \le k$ the number of rank $j$ matrices in
$\mbox{ker}(f)$  is
$$\frac{1}{q} \ \sum_{s=0}^k (-1)^{j-s} \ q^{ms+ \binom{j-s}{2}}\ 
{{k-s} \brack {k-j}} \left( {k \brack s} + (q-1) {k-R_f \brack s}\right).$$
\end{corollary}

 Let e.g. $f: \mbox{Mat} \to \F_q$ be the linear map that sends a matrix
to the sum of its entries. The orthogonal code of
$\mbox{ker}(f)$ is generated by the matrix whose entries are all ones, which has rank one.
By Corollary \ref{kerf},  for all $0 \le j \le k$ the number of rank $j$ matrices 
over $\F_q$
of size $k \times m$ whose entries sum to zero is
$$\frac{1}{q} \ \sum_{s=0}^k (-1)^{j-s} \ q^{ms+ \binom{j-s}{2}}\ 
{{k-s} \brack {k-j}} \left( {k \brack s} + (q-1) {k-1 \brack s}\right).$$
Generalizing the previous argument one obtains the following.

\begin{corollary} \label{lastfo}
 Let $I \subseteq [k] \times [m]$ be a non-zero set
of indices. For all $0 \le j \le k$ the number of $k \times m$ rank $j$
matrices 
$M$ over $\F_q$ such that $\sum_{(s,t) \in I} M_{st}=0$ is
$$\frac{1}{q} \ \sum_{s=0}^k (-1)^{j-s} \ q^{ms+ \binom{j-s}{2}}\ 
{{k-s} \brack {k-j}} \left( {k \brack s} + (q-1) {k-\mbox{rk} (M(I)) \brack
s}\right),$$
where $M(I)$ denotes the $k \times m$ matrix defined,
for all $(s,t) \in [k] \times [m]$, by
$M(I)_{st}:= 1$ if $(s,t) \in I$, and 
$M(I)_{st}:= 0$ otherwise.
\end{corollary}

The computation of the number of 
matrices over $\F_q$ with given size, rank 
and 
zero entries in a prescribed
region  is an active research area in combinatorics and combinatorial statistics
(see, among others \cite{haglund}, \cite{counting}, \cite{molti}, \cite{stemb} and the
references therein).
Such matrices can be regarded as $q$-analogues of permutations with restricted
positions.
It turns out that some instances of this type of enumeration problems can be investigated
using
MacWilliams identities
for the rank support, as we now show.

Let us first fix a convenient notation. The complement of
a set $I \subseteq [k] \times [m]$ is denoted by $I^c$.
For $I \subseteq [k] \times [m]$ define
$\mbox{Mat}[I]:= \{ M \in \mbox{Mat}: M_{st}=0 \mbox{ for all } 
(s,t) \in  I^c \}$.
Clearly, $\mbox{Mat}[I]$ is an $\F_q$-subspace of $\mbox{Mat}$ of dimension
$|I|$.

\begin{remark} \label{facileok}
For any subset $I \subseteq [k] \times [m]$
we have $\mbox{Mat}[I]^\perp = \mbox{Mat}[I^c]$.
Therefore by Theorem \ref{recallmw} the rank distributions of
$\mbox{Mat}[I]$ and $\mbox{Mat}[I^c]$ determine each other.
\end{remark}

For some sets $I$, the rank distribution
of $\mbox{Mat}[I]$ can be explicitly computed.
In these cases Theorem~\ref{recallmw} gives a formula
for the number of matrices in $\mbox{Mat}$ of any rank and zero entries on $I$.

\begin{corollary} \label{Yo}
 Let $1 \le k' \le k$ and $1 \le m' \le m$ be integers. For all
$0 \le j \le k$ the number of $k \times m$ rank $j$ matrices $M$ over $\F_q$
such that $M_{st}=0$ for all $(s,t) \in [k'] \times [m']$
is
$$q^{-k'm'} \ \sum_{i=0}^{\min\{ k',m'\}} 
{m' \brack i} \ \prod_{u=0}^{i-1} (q^{k'}-q^u)
 \sum_{s=0}^k (-1)^{j-s} \
q^{ms+\binom{j-s}{2}}  \ 
{{k-s} \brack {k-j}} {{k-i} \brack {s}}.$$
\end{corollary}

\begin{proof}
 Let $I:=[k'] \times [m']$. The code $\mC:=\mbox{Mat}[I]$
is the set of matrices whose entries are contained in the
rectangular region
described by $I$. As a consequence, for all  $0 \le i \le
\min\{k',m'\}$,
 $W_i(\mC,\omega_{\textnormal{rk}})$ is the number of  $k' \times
m'$ matrices
 over $\F_q$ with rank $i$, i.e,
$$W_i(\mC,\omega_{\textnormal{rk}})= {m' \brack i} \ \prod_{u=0}^{i-1}
(q^{k'}-q^u)
\ \ \mbox{ for $0 \le i \le \min \{ k',m'\}$.}$$
For $\min \{ k',m'\} < i \le k'$ we have
$W_i(\mC,\omega_{\textnormal{rk}})=0$. 
The result immediately follows from Remark \ref{facileok} and Theorem
\ref{recallmw}.
\end{proof}

Up to a permutation of rows and columns, the matrices of Corollary \ref{Yo} 
 have all their non-zero entries contained in a Ferrers board.
Matrices with this property have been widely studied in the literature 
(see \cite{haglund} among others).

Again concerning matrices with prescribed zero entries, a question of R. Stanley
asks for the number of invertible matrices 
over $\F_q$ having zero diagonal entries (see the Introduction of \cite{molti}).
The question was answered in \cite[Proposition 2.2]{molti}, where the authors
provide a formula for the number of $k \times m$ full-rank matrices over $\F_q$
with zero diagonal entries. Notice that for diagonal
entries
of a rectangular  matrix $M$  we mean the entries of the form
$M_{ss}$ for
$1 \le s \le k$.

 The following corollary
generalizes Proposition 2.2 of \cite{molti} with a simple proof 
based on MacWilliams identities.

\begin{corollary} \label{diaa}
 Let $I \subseteq \{ (s,t) \in [k] \times [m] : s=t \}$ be a set of diagonal
entries.
For all $0 \le j \le k$ the number of $k \times m$ matrices $M$ over $\F_q$
having rank $j$
and $M_{st}=0$ for all $(s,t) \in I$ is
$$q^{-|I|} \ \sum_{i=0}^{|I|} 
\binom{|I|}{i} (q-1)^i
 \sum_{s=0}^k (-1)^{j-s} \
q^{ms+\binom{j-s}{2}}  \ 
{{k-s} \brack {k-j}} {{k-i} \brack {s}}.$$
\end{corollary}

\begin{proof}
 Define $\mC:=\mbox{Mat}[I]$. For $|I| < i \le k$ we have $W_i(\mC,\omega_{\textnormal{rk}})=0$, and
for $0 \le i \le |I|$ we have 
$$W_i(\mC,\omega_{\textnormal{rk}})=\binom{|I|}{i} (q-1)^i.$$
Therefore the
formula  follows
from Remark \ref{facileok} and Theorem \ref{recallmw}.
\end{proof}

We conclude this section mentioning a very concise method to compute the
number of
symmetric and skew-symmetric $k \times k$ matrices of given rank over $\F_q$.
Different formul{\ae} for the same numbers were given by Carlitz in
\cite{carlitz1} and \cite{carlitz2} and by
 MacWilliams in \cite{mworth} using quite involved recursive arguments.
Our technique employs the M\"{o}bius inversion formula and the 
 regularity of the lattice of subspaces of $\F_q^k$, which we denote by $\mL$
in the sequel (see Example \ref{rwh}).

 Recall that a $k \times k$ matrix $M$ is \textbf{symmetric} if $M_{ij}=M_{ji}$ for all $1 \le i,j \le k$
  and \textbf{skew-symmetric} if $M_{ii}=0$ and $M_{ij}=-M_{ji}$ for all $1 \le
 i,j \le k$. 
  We denote by
 $\mbox{Sym}$ and $\mbox{s-Sym}$ the spaces of
 $k \times k$ symmetric and skew-symmetric matrices over $\F_q$,
 respectively.

\begin{lemma} \label{cg}
 Let $S \subseteq \F_q^k$ be any $s$-dimensional subspace. Then
  $\{M \in \mbox{Sym} : \sigma_{\textnormal{rk}}(M) \subseteq S\}$ has dimension
 $s(s+1)/2$ over $\F_q$. 
\end{lemma}

\begin{proof}
 Define $V:=\{
x \in \F_q^k : x_i=0 \mbox{ for } i >s \} \subseteq \F_q^k$. There exists an
$\F_q$-isomorphism $g:\F_q^k \to \F_q^k$ such that $g(S)=V$.
Let $G \in
\mbox{GL}_k(\F_q)$ be the matrix associated to $g$ with
respect to the canonical
basis $\{e_1,...,e_k\}$ of $\F_q^k$.
Since $G$ is invertible, the map $M \mapsto GMG^t$ is an $\F_q$-linear isomorphism 
$\{M \in \mbox{Sym} :
\sigma_{\textnormal{rk}}(M) \subseteq S\} \to 
\{M \in \mbox{Sym} : \sigma_{\textnormal{rk}}(M) \subseteq V\}$.
The lemma now follows from the fact that 
$\dim (\{M \in \mbox{Sym} : \sigma_{\textnormal{rk}}(M) \subseteq V\})=s(s+1)/2$.
\end{proof}

We can now compute the number  of symmetric 
 $k \times k$ matrices over $\F_q$ of rank $i$ as follows.
For any subspace $T \subseteq \F_q^k$ define
$f(T):= |\{ M \in \mbox{Sym} : \sigma_{\textnormal{rk}}(M)=T\}|$ and
$g(T):= \sum_{ S \subseteq T}  f(S)$.
By Lemma \ref{cg}, for all $S \subseteq \F_q^k$ we have
$g(S)=q^{s(s+1)/2}$, where $s:= \dim(S)$.
Therefore applying the M\"{o}bius inversion formula  (\cite{ec},
Proposition
3.7.1) to the functions $f$ and $g$ we obtain, for
any given $i$-dimensional subspace $T \subseteq \F_q^k$,
$$
 f(T) \ = \ \sum_{\substack{S \subseteq  T}} g(S) \ \mu_\mL(S,T) \ = \
\sum_{s=0}^k \sum_{\substack{S \subseteq T \\ \dim(S)=s}}
q^{s(s+1)/2} \ \mu_\mL(s,i) \ = \ 
 \sum_{s=0}^k
q^{\binom{s+1}{2}} {i \brack s} (-1)^{i-s} q^{\binom{i-s}{2}}.
$$
The expected result is now derived summing over all the $i$-dimensional
subspaces $T \subseteq \F_q^k$.
A similar argument applies to skew-symmetric matrices. The final result is the following.

\begin{proposition}
 The number of symmetric  and skew-symmetric 
 $k \times k$ matrices over $\F_q$ of rank $i$ is, respectively, 
$${k \brack i} \ 
 \sum_{s=0}^k (-1)^{i-s}
 q^{\binom{s+1}{2} + \binom{i-s}{2}} {i \brack s}, \ \ \ \ \ \ \ 
 \ \ \ \ \ \ {k \brack i} \ 
 \sum_{s=0}^k (-1)^{i-s}
 q^{\binom{s}{2} + \binom{i-s}{2}} {i \brack s}.$$
\end{proposition}

One can also observe that the spaces of $k \times k$ symmetric and 
skew-symmetric matrices 
over $\F_q$ are orthogonal to each other. Therefore the rank 
distributions of symmetric and skew-symmetric matrices
are related by a MacWilliams transformation. More precisely, the following hold.

\begin{corollary} \label{finalss}
  For all integers $0 \le j \le k$ we have
  $$W_j(\mbox{Sym},\omega_{\textnormal{rk}})={q^{-\binom{k}{2}}}
 \sum_{i=0}^k
 W_i(\mbox{s-Sym},\omega_{\textnormal{rk}})
  \sum_{s=0}^k (-1)^{j-s} \
 q^{ks+\binom{j-s}{2}}  \ 
 {{k-s} \brack {k-j}} {{k-i} \brack {s}}.$$
 \end{corollary}
 

\section*{Acknowledgement}
The author is grateful to Elisa Gorla, Frank R. Kschischang, and the 
Referees of this paper for help in improving Section
\ref{recover} and the  
presentation of this work.


\vspace{4em}

\begin{thebibliography}{55}

\bibitem{andr} G. E. Andrews, \emph{The Theory of Partitions}. Encyclopedia
of Mathematics and its Applications, vol. 2,  G.C. Rota Editor. Addison-Wesley,
1976.

 \bibitem{B90} P.G. Bonneau, {\em Weight Distributions of Translates of MDS Codes}, 
Combinatorica, 10 (1), 103--105, 1990.

\bibitem{eb} E. Byrne, \emph{On the weight distribution of codes over 
finite rings}. Advances in Mathematics of Communications, 5 (2011), pp. 
395 --406.


\bibitem{by} E. Byrne, M. Greferath, M. E. O'Sullivan, \emph{The linear
programming bound for codes over finite Frobenius
rings}. Designs, Codes and Cryptography, 42 (2007), pp. 289 -- 301.



 \bibitem{camion} P. Camion, \emph{Codes and association schemes}. In V. S.
 Pless 
 and W. C. Huffman (editors), Handbook of Coding Theory, Vol. II, pp. 1441 --
 1566. Elsevier (1998).


\bibitem{carlitz1} L. Carlitz, \emph{Representations by quadratic forms in a
finite field}. Duke Mathematical Journal, 21 (1954), pp. 123 -- 137.

\bibitem{carlitz2} L. Carlitz, \emph{Representations by skew forms in a finite
field}. Archiv der Mathematik, 5 (1954), pp. 19 -- 31.


\bibitem{cost} I. Constantinescu, W. Heise, \emph{A metric for codes over 
residue class rings}. Problems on Informormation Transmission, 33 (1997), pp. 208 --213.


\bibitem{del2} P. Delsarte, \emph{Association schemes and $t$-designs
in regular semilattices}. Journal of Combinatorial Theory, Series A, 2
(1976), 2, pp. 230 -- 243.

\bibitem{del1} P. Delsarte, \emph{Bilinear forms over a finite field, with
applications to coding theory}. Journal of Combinatorial Theory, Series A, 25
(1978), 3, pp. 226 -- 241.

\bibitem{delphil} P. Delsarte, \emph{An algebraic approach to the association
schemes of coding theory}.
Philips Research Report, Supplement, 10 (1973).

\bibitem{forney} G. D. Forney, Jr, \textit{Transforms and groups}. In A. 
Vardy (ed.), \textit{Codes, Curves and Signals:
Common Threads in Communications}, pp. 79 -- 97. Kluwer, 1998.

\bibitem{heide1} H. Gluesing-Luerssen, \emph{Fourier-reflexive partitions and
MacWilliams identities for additive codes}. Designs, Codes and Cryptography, 75
(2015), pp. 543 -- 563.


\bibitem{heide2} H. Gluesing-Luerssen,
\emph{Partitions of Frobenius rings induced by the homogeneous weight}.
Advances in Mathematics of Communications, 8 (2014), pp. 191 -- 207.


\bibitem{gref} M. Greferath, S. Schmidt, \emph{Finite ring 
combinatorics and MacWilliams' Equivalence Theorem}.
Journal of Combinatorial Theory,  92A (2000), pp. 17 -- 28.


\bibitem{haglund} J. Haglund, \emph{$q$-rook polynomials and matrices over 
finite fields}. Advances in Applied Mathematics, 20 (1998), 4, pp. 
450 -- 487. 


\bibitem{z4} A. R. Hammons, P. V. Kumar, A. R. Calderbank, N. J. A. Sloane,
P. Sol\'{e}, \emph{The $\Z_4$-linearity of Kerdock, Preparata, Goethals,
and related codes}. IEEE Transactions on Information Theory,
40 (1994), pp. 301 -- 319.

 
 \bibitem{hon} T. Honold, I. Landjev, \emph{MacWilliams identities for 
linear codes over finite Frobenius rings}. In D. Jungnickel and H. Niederreiter, editors,
 \emph{Proceedings of The Fifth International Conference on Finite Fields and Applications Fq5},
  Augsburg, 1999, pp. 276 -- 292. Springer 2001.
 
 \bibitem{pless} W. C. Huffman, V. Pless, \emph{Fundamentals of Error-Correcting
 Codes}. Cambridge University Press (2003).
 
 \bibitem{counting}  A. J. Klein, J. B. Lewis, A. H. Morales,
 \emph{Counting matrices over finite fields with support on skew Young diagrams and 
 complements of Rothe diagrams}. 
 Journal of Algebraic Combinatorics, 39 (2014), 2, pp. 429 -- 456. 
 
 
 
 \bibitem{lam} T. Y. Lam, \emph{Lectures on Modules and Rings}. 
 Graduate Text in Mathematics, vol. 189. Springer 1999.
 
 
 \bibitem{leew} C. Lee, \emph{Some properties of nonbinary error-correcting codes}.
 IRE Transactions on Information Theory, 4 (1958), 2,
 pp. 77 -- 82. 
 
 
 \bibitem{molti} J. B. Lewis, R. Liu, G. Panova, A. H.  Morales, S. V Sam, 
 Y. X.  Zhang, \emph{Matrices with restricted entries
  and $q$-analogues of permutations}. Journal of Combinatorics, 2 (2012), 3, pp. 355 -- 396.
 
 
 \bibitem{vl} J. H. van Lint, \emph{Introduction to Coding Theory}, third
 edition. Springer (1999).


 

 
 \bibitem{ooo} F. J. MacWilliams, \emph{A Theorem on the Distribution of Weights in a
 Systematic Code}. Bell  System Technical Journal, 42 (1963), 1, pp. 79 -- 94.
 
 
 \bibitem{mworth} F. J. MacWilliams, \emph{Orthogonal matrices over finite
 fields}. American Mathematical Monthly, 76 (1969), pp. 152 -- 164.
 
 \bibitem{MS} F. J. MacWilliams, N. J. A. Sloane, \emph{The Theory of
  Error-Correcting Codes}. North Holland Mathematical Library (1977).


 \bibitem{al} A. Ravagnani, \emph{Rank-metric codes and their duality theory}.
 Designs, Codes and Cryptography, 80 (2016), 1, pp. 197 -- 216.
 
 
 

\bibitem{metrics} D. Silva, F. R. Kschischang, \emph{On metrics for error
correction in network coding}.
IEEE Transactions on Information Theory, 55 (2009), 12, pp. 5479 -- 5490. 
 
 
  \bibitem{incal} E. Spiegel, C. J. O'Donnell, \emph{Incidence algebras}. CRC
 Press (1997).
 
 \bibitem{ec} P. Stanley, \emph{Enumerative Combinatorics}, vol. 1, second
 ed., Cambridge Stud. Adv. Math., vol. 49, Cambridge University Press,
 Cambridge (2012).
   
  \bibitem{stemb} J. R. Stembridge, \emph{Counting points on varieties over 
 finite fields related to a conjecture of
 Kontsevich}. Annals of Combinatorics, 2 (1998), 4, pp. 365 -- 385. 
 
 
 
 \bibitem{zino0} V. A. Zinoviev, T. Ericson,
 \emph{On Fourier invariant partitions of finite abelian groups and
the MacWilliams identity for group codes}.
Problems of Information Transmission, 32 (1996), pp. 117 -- 122.
 
   \bibitem{zino} V. A. Zinoviev, T. Ericson, \emph{Fourier invariant pairs of
   partitions of finite abelian groups and association
   schemes}. Problems of Information Transmission, 45 (2009), pp. 221 --
   231.
   



\end{thebibliography}
\end{document}